\documentclass[journal,final,twoside]{IEEEtran}

\usepackage[T1]{fontenc}

\usepackage{graphicx}
\usepackage{tikz}
\usetikzlibrary{svg.path}
\usepackage{scalerel}
\usepackage{cite}
\usepackage{amsmath}
\usepackage{floatrow}
\usepackage{eso-pic}
\usepackage{amsthm}
\usepackage{mathtools}
\usepackage{psfrag}
\usepackage{dsfont}

\newtheorem{theorem}{Theorem}
\newtheorem{lem}{Lemma}
\newtheorem{cor}{Corollary}

\usepackage{amsmath}
\usepackage[cmintegrals]{newtxmath}

\newcommand{\SNR}{\mathrm{SNR}}
\newcommand{\SINR}{\mathrm{SINR}}

\newcommand{\rmin}{r_{\mathrm{min}}}
\newcommand{\rmax}{r_{\mathrm{max}}}
\newcommand{\rEarth}{r_{\oplus}}

\newcommand{\Prob}{{\mathbb{P}}}

\newcommand{\Pc}{P_\mathrm{c}}
\newcommand{\PI}{P_\mathrm{I}}

\newcommand{\ps}{p_\mathrm{s}}
\newcommand{\ppi}{p_\mathrm{i}}
\newcommand{\phiu}{\phi_\mathrm{u}}
\newcommand{\phii}{\phi_\mathrm{i}}

\newcommand{\NI}{N_\mathrm{I}}
\newcommand{\nI}{n_\mathrm{I}}
\newcommand{\Neff}{N_\mathrm{eff}}

\definecolor{orcidlogocol}{HTML}{A6CE39}
\tikzset{
  orcidlogo/.pic={
    \fill[orcidlogocol] svg{M256,128c0,70.7-57.3,128-128,128C57.3,256,0,198.7,0,128C0,57.3,57.3,0,128,0C198.7,0,256,57.3,256,128z};
    \fill[white] svg{M86.3,186.2H70.9V79.1h15.4v48.4V186.2z}
                 svg{M108.9,79.1h41.6c39.6,0,57,28.3,57,53.6c0,27.5-21.5,53.6-56.8,53.6h-41.8V79.1z M124.3,172.4h24.5c34.9,0,42.9-26.5,42.9-39.7c0-21.5-13.7-39.7-43.7-39.7h-23.7V172.4z}
                 svg{M88.7,56.8c0,5.5-4.5,10.1-10.1,10.1c-5.6,0-10.1-4.6-10.1-10.1c0-5.6,4.5-10.1,10.1-10.1C84.2,46.7,88.7,51.3,88.7,56.8z};
  }
}

\newcommand\orcidicon[1]{\href{https://orcid.org/#1}{\mbox{\scalerel*{
\begin{tikzpicture}[yscale=-1,transform shape]
\pic{orcidlogo};
\end{tikzpicture}
}{|}}}}

\usepackage{hyperref} 

\title{Downlink Coverage and Rate Analysis of\\Low Earth Orbit Satellite Constellations\\Using Stochastic Geometry}
\author{\IEEEauthorblockN{Niloofar Okati \orcidicon{0000-0002-8074-5146}, Taneli Riihonen \orcidicon{0000-0001-5416-5263},~\IEEEmembership{Member,~IEEE,} Dani Korpi \orcidicon{0000-0003-3460-7436}, Ilari Angervuori, and Risto Wichman}%
\thanks{Manuscript received October 30, 2019; revised March 4, 2020; accepted April 17, 2020.
The work of N.~Okati and T.~Riihonen was supported by a Nokia University Donation.
The associate editor coordinating the review of this paper and approving it for publication was S.~Durrani. {\em (Corresponding author: Niloofar Okati.)}}%
\thanks{N.~Okati and T.~Riihonen are with Unit of Electrical Engineering, Faculty of Information Technology and Communication Sciences, Tampere University, FI-33720 Tampere, Finland (e-mail: niloofar.okati@tuni.fi; taneli.riihonen@tuni.fi).}%
\thanks{D. Korpi is with Nokia Bell Labs, Karaportti 3, FI-02610 Espoo, Finland (e-mail: dani.korpi@nokia-bell-labs.com).}%
\thanks{I. Angervuori and R. Wichman are with Department of Signal Processing and Acoustics, Aalto University School of Electrical Engineering, FI-00076 Espoo, Finland (e-mail: ilari.angervuori@aalto.fi; risto.wichman@aalto.fi).}
\thanks{Digital Object Identifier 10.1109/TCOMM.2020.2990993}
}

\begin{document}

\markboth{IEEE Transactions on Communications}%
{Okati \MakeLowercase{\textit{et al.}}: Downlink Coverage and Rate Analysis of LEO Satellite Constellations Using Stochastic Geometry}

\maketitle

\begin{abstract}
As low Earth orbit (LEO) satellite communication systems are gaining increasing popularity, new theoretical methodologies are required to investigate such networks' performance at large. This is because deterministic and location-based models that have previously been applied to analyze satellite systems are typically restricted to support simulations only.
In this paper, we derive analytical expressions for the downlink coverage probability and average data rate of generic LEO networks, regardless of the actual satellites' locality and their service area geometry. Our solution stems from stochastic geometry, which abstracts the generic networks into uniform binomial point processes. Applying the proposed model, we then study the performance of the networks as a function of key constellation design parameters. Finally, to fit the theoretical modeling more precisely to real deterministic constellations, we introduce the effective number of satellites as a parameter to compensate for the practical uneven distribution of satellites on different latitudes.
In addition to deriving exact network performance metrics, the study reveals several guidelines for selecting the design parameters for future massive LEO constellations, e.g., the number of frequency channels and altitude.
\end{abstract}


\IEEEpeerreviewmaketitle

\section{Introduction} 
\IEEEPARstart{T}{he} challenge for providing affordable Internet coverage everywhere around the world requires novel solutions for ubiquitous connectivity. An emerging technology, which can provide the infrastructure with relatively low propagation delay compared to conventional geostationary satellites and seamless connectivity also at polar regions, is massive low Earth orbit (LEO) satellite networking. Many LEO constellations, e.g., Kuiper, LeoSat, OneWeb, Starlink, Telesat, etc., have secured investors or, even, are already launching pilot satellites. While commercial plans obviously must have been simulated thoroughly before putting forward, general theoretical understanding on the performance of massive LEO communication constellations at large is still missing in the scientific literature.

In this paper, we apply stochastic geometry to acquire analytical tractable expressions for coverage probability and average data rate of downlink LEO networks. The approach formulated herein for the first time ever to the authors' knowledge paves the way to study the generic performance of satellite networking without relying on explicit orbit simulations and the actual geometry of any specific constellation.

\subsection{Related Works}
In \cite{1}, the uplink outage probability in the presence of interference was evaluated for two LEO constellations through time-domain simulations. A performance study of Iridium constellation was presented in \cite{4} in terms of system capacity, the average number of beam-to-beam handoffs and satellite-to-satellite handoffs, the channel occupancy distribution and average call drop probability.
The effect of traffic non-uniformity on signal-to-interference ratio was studied in \cite{38} by assuming hexagonal service areas for satellites. 
A teletraffic analysis of a mobile satellite system based on a LEO constellation was performed in \cite{5}. A general expression for a single LEO satellite's visibility time is provided in \cite{8}, but it is incapable of concluding the general distribution of visibility periods for any arbitrarily positioned user.

\begin{figure*}[t]
    \includegraphics[trim=1.5cm 0.7cm 1.5cm 17cm,width=\textwidth]{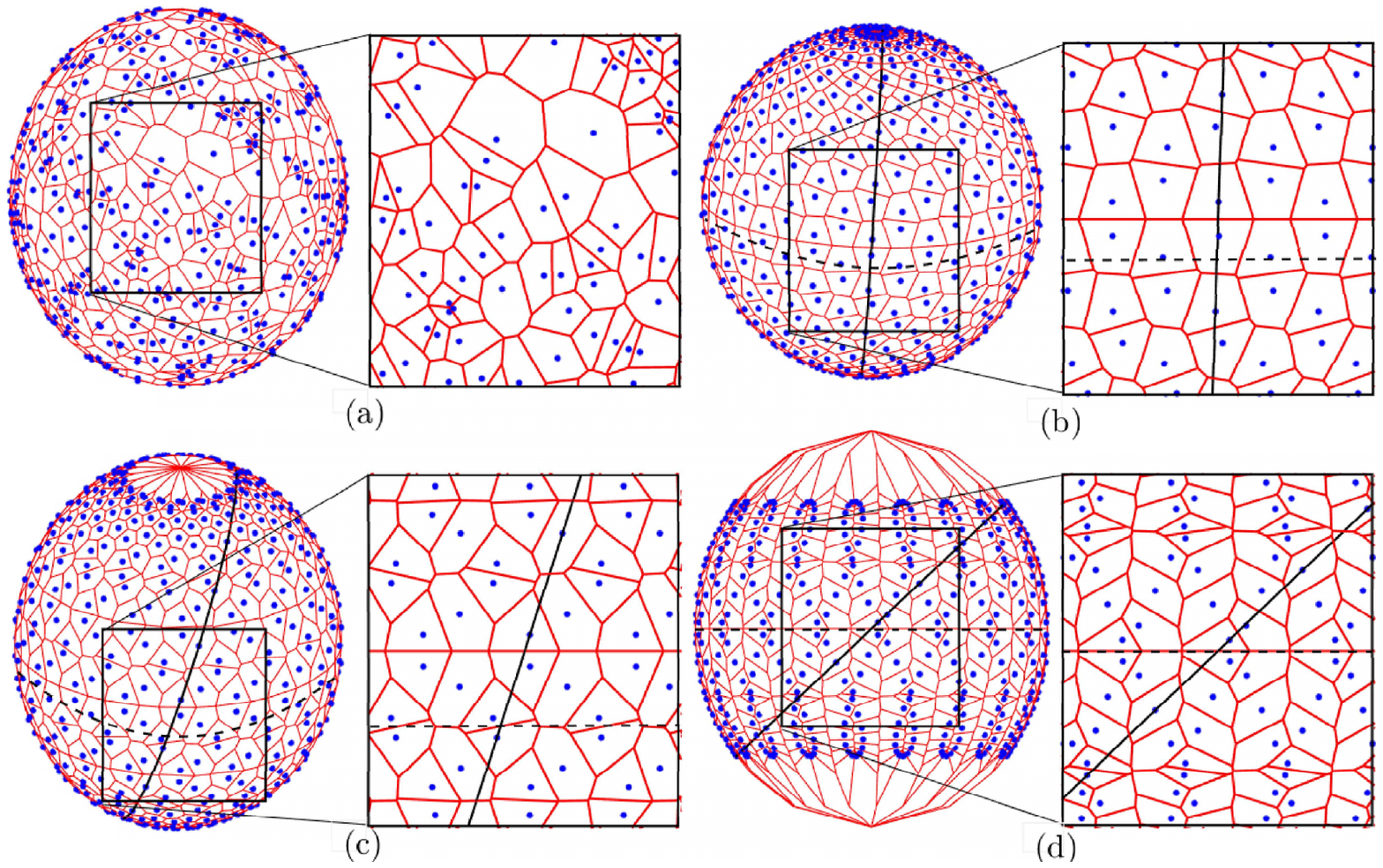}
    \caption{{Example orbits and} spherical Voronoi diagrams representing the coverage areas of the nearest satellites with (a) a random constellation, (b) a polar constellation with $87.9^{\circ}$ inclination angle, (c) an inclined constellation with $70^{\circ}$ inclination angle, and (d) an inclined constellation with $40^{\circ}$ inclination angle.}
    \label{fig:spherical_Voronoi_examples}
\end{figure*}

Stochastic geometry is a powerful mathematical
and statistical tool for the modeling, analysis, and design of wireless networks with irregular topologies \cite{24,blaszczyszyn2018stochastic,25}. It has been mostly used in the literature to analyze two-dimensional (planar) terrestrial networks  \cite{23,24,25,9,10,11,12,28,blaszczyszyn2018stochastic}. In \cite{23}, a comprehensive review of the literature related to
the stochastic geometry modeling of multi-tier and cognitive cellular networks was presented. {In \cite{9}, the Poisson point process (PPP) was observed to provide lower bounds on the coverage probability and the average transmission rate which are as tight as an upper bound provided by an idealized grid-based model.
The work in \cite{9} was extended to multi-tier networks under quality-of-service constraints in \cite{10, 11}. Uplink coverage when base stations and users follow independent PPPs is evaluated in \cite{12}.}

While stochastic geometry has been extensively applied for the analysis of planar scenarios, {its application to three-dimensional networks has started attracting significant attention recently \cite{19,20}.}
In \cite{19}, modeling and analysis of coverage in three-dimensional cellular networks have been investigated using a PPP model.
{Although PPP has provided tractable and insightful results, it is not valid for modeling a finite-area network with a limited number of nodes \cite{41}. For such cases, binomial point process (BPP) is an appropriate model to capture the characteristics of the network \cite{42,43}. The performance for an arbitrarily shaped planar network is studied in \cite{42,44}, but the reference transmitter is assumed to be located at a given distance, which can be overcome with two transmitter selection policies from \cite{28}. A BPP modeling for a finite three-dimensional network of unmanned aerial vehicles (UAVs) was developed in \cite{20,45}.}

Tools from stochastic geometry have been applied to study satellite communications \cite{30,39,31,6} only with a limited extent. {For instance, in \cite{30}, stochastic geometry was used to model the locations of terrestrial users involving multi-beam satellite and terrestrial interference from cellular base stations.} The authors in \cite{39} derived coverage probability and data rate of a multi-UAV downlink network through assuming PPP distribution for users where, however, no interference among users and UAVs is considered due to their channel assignment policy. In \cite{31}, the performance of a cognitive satellite--terrestrial network is investigated; the secondary terrestrial network and users are modeled as independent point processes and share resources with a primary satellite system. An analysis of coverage times during LEO satellite visits has been conducted in \cite{6} by inclusion of the distribution of users' positions.

We can conclude from the above that the utilization of stochastic geometry for satellite network analysis is limited to modeling the user's locality with a low number of satellites. Alternative methods to stochastic geometry are unable to provide a comprehensive analysis that fits to any arbitrary constellation due to being associated with some specific network design parameters. In these models, coverage footprints of satellites are assumed to be identical and typically form a regular circular or hexagonal grid, although uneven distribution of satellites along different latitudes, and differences in transmitted power and/or altitude, create irregular cells. Moreover, there is no analytical tool to model interference in a generic sense.

\subsection{Contributions and Organization of the Paper}

{ Figure~1(a) demonstrates a random constellation in which a set of satellites is distributed on a sphere according to a uniform point process, while Figs.~1(b)--(d) show a class of regular deterministic Walker constellations with varying inclination angle, in which all satellites are evenly spaced and have the same period and inclination. Each cell represents the coverage area of a satellite wherein it is the closest to and, thus, serving all users located inside the polygon. The set of all coverage cells form a Voronoi tessellation which we can observe to be analogous in the zoomed regions for the random constellation and deterministic Walker constellations.
While selecting the inclination angle in practice is based on the service area of interest, we see that smaller inclination results in a more irregular Voronoi tessellation.
Furthermore, the inter-operation of multiple LEO constellations, as it might hold in the near future, will make satellites' mutual positions even more similar to those given by a random process.}

Motivated by the above observations, in this paper, we first model the satellite constellation with an appropriate point process, which will then allow us to utilize the tools from stochastic geometry to analyze the performance of generic LEO networks in theory. We consider a network of a given number of satellites whose locations are modeled as a BPP on a sphere at a fixed altitude, which is justifiable due to the limited number of satellites covering a given finite region~\cite{34}.  
Users are located at some arbitrary locations on Earth and are associated with the nearest satellite while some other satellites above the horizon of a user can cause co-channel interference due to frequency reuse. {However, since satellites in Walker constellations are distributed unevenly along different latitudes, i.e., the number of satellites is effectively larger on the inclination limit of the constellation than on the equatorial regions, the density of practical deterministic constellations is typically not uniform like with BPP modeling. Thus, we apply a new parameter in order to compensate for the uneven density w.r.t.\ practical Walker constellations and create a tight match between the results generated by BPP modeling and those from practical constellation simulations.}

Based on the modeling summarized above, we present the following scientific contributions in this paper.
\begin{itemize}
	
\item 
We derive exact expressions for coverage probability and average achievable data rate of a user in terms of the Laplace transform of interference power distribution.

\item 
We validate our novel theoretical results with numerical simulations and also compare them with reference results from actual deterministic satellite constellations.
 
\item 
To suppress the performance mismatch between a random network and practical constellations, which mostly stems from uneven distribution of satellites along different latitudes, we define and calculate a new parameter, the effective number of satellites, for every user latitude.  

\item {We show that, with the above compensation, the generic performance of large deterministic constellations can be very accurately analyzed with theoretical expressions that are based on stochastic constellation geometry.}

\item Finally, the two objectives of coverage and data rate are evaluated for different key design parameters, e.g., the number of frequency channels and satellite altitude. 
\end{itemize}

 As for propagation models, we consider two extreme cases, namely Rayleigh fading and static propagation, for the serving channels. The former corresponds to a more drastic fading environment when the received signal is subject to severe multi-path distortion due to the small elevation angle of the transmitting satellite; this case leads to simpler expressions for some specific path loss exponents. The latter represents the typical cases where the serving satellite's elevation is large enough to provide line-of-sight to the user and, thus, weak components from multi-path propagation become insignificant. For interfering channels, any fading statistics can be adopted in general since this has no effect on analytical tractability.

We see that, although increasing the number of frequency bands improves the coverage probability, { there is an optimal number of frequency channels that maximizes the data rate depending on the path loss exponent. Assuming a constant effective number of satellites for different altitudes,} we observe that the optimum height which maximizes coverage probability or data rate is not within the practical altitude range of LEO networks (i.e., outside Earth's atmosphere), where the performance always declines with increasing the altitude.

The organization of the remainder of this paper is as follows. Section~II describes the system model for a randomly distributed satellite network and characterizes some baseline probabilities stemming from the stochastic geometry of the system. As for the main results, we derive analytical expressions for downlink coverage probability and average achievable data rate for a terrestrial user in Sections~III and IV, respectively. Numerical results are provided in Section~V for studying the effect of key system parameters such as satellite altitude and the number of frequency bands allocated for the network. Finally, we conclude the paper in Section~VI.

\begin{figure}[t]
 \includegraphics[trim=6.2cm 11cm 6.2cm 12.7cm,width=\textwidth]{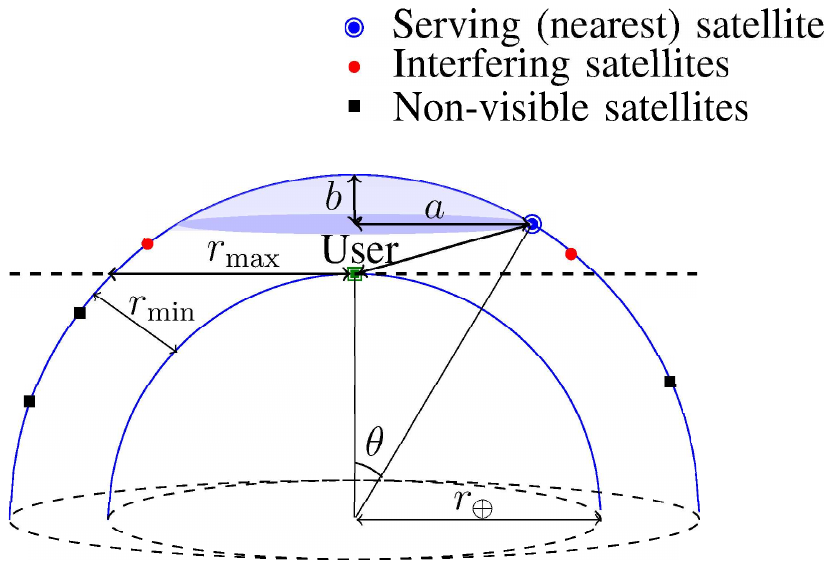}
\caption{A sketch of the considered system's stochastic geometry, where satellites are distributed randomly on a sphere with radius $\rEarth+\rmin$ and a user is located on the surface of another cocentric sphere with radius $\rEarth$.}
\label{fig:1} 
\end{figure}
\section{System Model}

Let us consider a downlink network of $N$ satellite base stations that are uniformly distributed around Earth at the same altitude $\rmin$ forming a BPP on a sphere with radius $\rEarth+\rmin$, where $\rEarth$ denotes Earth's radius, as shown in Fig.~2, while user terminals are located on the surface of Earth. The altitude parameter $\rmin$ specifies also the minimum possible distance from a satellite to a user (that is realized when it is directly above the user), hence the name. We assume that wireless transmissions propagate to a user from all and only the satellites that are above its horizon. Correspondingly, $\rmax=\sqrt{2\rEarth\rmin+\rmin^2}$ denotes the maximum possible distance at which a satellite has any effect on the network service to a user (that is realized when the satellite is at the horizon). The notation followed in this paper is summarized in Table~I.

Each user is associated to the nearest satellite that is referred to as the serving satellite in what follows, resulting in spherical Voronoi tessellation for satellites' coverage areas as illustrated in Fig.~\ref{fig:spherical_Voronoi_examples}(a). The other satellites may cause co-channel interference to the user because we assume that there are only $K$, with $K\leq N$, orthogonal frequency channels for the network and randomly assign a subset of $N/K$ satellites to each channel. The scheduling performed in this way ensures that the nearest satellite in the constellation uses the channel that is assigned to a user. All the other satellites on the same channel, other than the serving satellite, cause interference to the user's reception whenever they are above the horizon.

\begin{table*}[!t]
\label{tab:1}
\begin{tabular}{c|cc}
\hline
\hline
\textbf{Notation}&\textbf{Description}\\
\hline
$\rEarth; \rmin $&Earth radius ($6371$ km); Altitude of satellites\\
\hline
$R_0; R_n; R$&Serving distance; Distance to the $n^{\mathrm{th}}$ interfering satellite; Distance from the user to any satellite\\
\hline
$G_0; G_n$&Channel fading gain of the serving link; Channel fading gain of the $n^{\mathrm{th}}$ interfering link\\
\hline
$N; K; \NI$&Total number of satellites; Total number of channels; Number of interfering satellites\\
\hline
$\ps $& Transmission power from the serving satellite\\
\hline
$\ppi $&Transmission power from interfering satellites \\
\hline
 $\sigma^2$&Additive noise power\\
\hline
 $\alpha$&Path loss exponent\\
\hline
 $T$&SINR threshold\\
\hline
 $\Pc; \bar{C}$&Coverage probability; Average achievable rate\\
\hline
\hline
\end{tabular}
\caption{Summary of mathematical notation}
\end{table*}
{We suppose that every satellite may be equipped with a directional antenna that radiates its main lobe towards the center of Earth. From the serving satellite, some lobe containing the higher power (but not necessarily the main lobe) is likely directed to the user, while the lobes with lower power levels radiate towards the user from the interfering satellites. In order to approximate the effect of directional transmission, we set differing power levels for serving and interfering satellites which are denoted by $\ps$ and $\ppi$, respectively, such that $\ppi\leq\ps$. If the user is located within the beamwidth of the serving satellite, the transmitted power from it will have the main lobe power level. Considering the main lobe transmitted power from the serving satellite and the largest sidelobe transmitted power from the interfering satellites, $\ps/\ppi$ corresponds to the sidelobe level of the antennas.}

The distances from the user to the serving satellite and the other satellites are denoted by random variables $R_0$ and $R_n$, \mbox{$n=1,2,\ldots,N-1$}, respectively, while $G_0$ and $G_n$ represent the corresponding channel gains. Obviously, $G_n=0$ if $R_n > \rmax$ for some $n=0,1,\ldots,N-1$, i.e., the satellite is below horizon. For notational convenience, when $\NI>0$, we let indices $n=1,2,\ldots,\NI$ correspond to those $\NI \leq N/K-1$ satellites (if any) that share the same frequency channel with the serving satellite and are above the user's horizon, so that they cause co-channel interference.

Based on the above modeling, the signal-to-interference-plus-noise ratio (SINR) at the receiver can be expressed as
\begin{align}
\label{eq:1}
\SINR = \frac{\ps G_0 R_0^{-\alpha}}{I+\sigma^2},
\end{align}
where we assume that the user's receiver is subject to additive noise with constant power $\sigma^2$, the parameter $\alpha$ is a path loss exponent,
\begin{align}
\label{eq:2}
I = \sum_{n=1}^{\NI} \ppi G_n R_n^{-\alpha}
\end{align}
is the cumulative interference power from all other satellites above the user's horizon than the serving satellite, and $\NI$ is a random variable denoting the number of interfering satellites. In the special case of having no interfering satellites, i.e., \mbox{$\NI = 0$}, since they happen to be below the user's horizon, the SINR in \eqref{eq:1} is reduced to signal-to-noise ratio (SNR) as
\begin{align}
\label{eq:41}
\SNR = \frac{\ps G_0 R_0^{-\alpha}}{\sigma^2},
\end{align}
and further $\SNR=0$ if $R_0 > \rmax$, i.e., also the nearest satellite that is supposed to be the serving satellite is below horizon.

In order to contribute expressions for coverage probability and average achievable rate in the following sections, we first need to characterize some basic distance distributions that stem from the stochastic geometry of the considered system. In particular, we express the necessary cumulative distribution function (CDF) and probability density functions (PDFs) in the following three lemmas.

\begin{lem}
The CDF of the distance $R$ from any specific one of the satellites in the constellation to the user is given by
\begin{equation}
\label{eq:9}
F_{R}\left(r\right) \triangleq \Prob\left(R \leq r\right)=\left\{
\begin{array}{ll}
0, & r<\rmin,\\
\frac{r^2-\rmin^2}{4\rEarth(\rEarth+\rmin)}, &\rmin\leq r \leq2\rEarth+\rmin,\\
1,&r>2\rEarth+\rmin,\\
\end{array} \right.
\end{equation}
and the corresponding PDF is given by  
\begin{equation}
\label{eq:10}
f_{R}\left(r\right) = \frac{r}{2\rEarth(\rEarth+\rmin)}
\end{equation}
for $\rmin\leq r \leq 2\rEarth+\rmin$ while $f_{R}\left(r\right)=0$ otherwise.
\end{lem}
\begin{proof}
See Appendix A.
\end{proof}

{The above distribution of $R$ will be a key tool in deriving other distance statistics (i.e., those of the serving and interfering satellites), which are required in turn to characterize the distribution of SINR in \eqref{eq:1} using stochastic geometry.}

\begin{lem}
The PDF of the serving distance $R_0$ is given by 
\begin{align}
\label{eq:11}
f_{R_0}\left(r_0\right) = N\left(1-\frac{r_0^2-\rmin^2}{4\rEarth(\rEarth+\rmin)}\right)^{N-1} \frac{r_0}{2\rEarth(\rEarth+\rmin)}
\end{align}
for $\rmin\leq r_0 \leq 2\rEarth+\rmin$ while $f_{R_0}\left(r_0\right)=0$ otherwise.
\end{lem}
\begin{proof}
Due to the channel assignment by which the serving satellite is the nearest one among all the $N$ i.i.d.\ satellites, the CDF of $R_0$ can be expressed as $F_{R_0}\left(r_\text{0}\right)\triangleq\Prob\left(R_0 \leq r_\text{0}\right) = 1 - \left(1-F_{R}\left(r_\text{0}\right)\right)^N$ with the substitution of (\ref{eq:9}). The corresponding PDF is obtained by differentiation to complete the proof.
\end{proof}

\begin{lem}
When conditioned on the serving distance such that $R_0=r_0$, the PDF of the distance from any other satellite to the user is given by
\begin{align}
\label{eq:14}
f_{R_n|R_0}\left(r_n|r_0\right)=\frac{f_{R}\left(r_n\right)}{1-F_{R}\left(r_0\right)}
\end{align}
for $r_0< r_n\leq 2\rEarth+\rmin$ while $f_{R_n|R_0}\left(r_n|r_0\right)=0$ otherwise.
\end{lem}
\begin{proof}
The CDF of $R_n|_{R_0=r_0}$ is obtained by conditioning $R$ on $R_0$ as follows:
\begin{align}
\label{eq:25}
\nonumber
F_{R_n|R_0}(r_n|r_0)&\triangleq\Prob\left(R_n<r_n|R_0=r_0\right)\\
&=\frac{\Prob\left(r_0 \leq R \leq r_n\right)}{\Prob(R>r_0)} = \frac{F_{R}(r_n)-F_{R}(r_0)}{1-F_{R}(r_0)}.
\end{align}

The proof is then completed by taking the derivative of the CDF with respect to $r_n$ for obtaining the PDF in (\ref{eq:14}).
\end{proof}

Furthermore, other auxiliary performance factors, which will be shortly used in analyzing coverage probability and average achievable data rate, are the number of interfering satellites $\NI$ and the probability $P_0$ of having an interference-free situation ($\NI=0$) when $R_0=r_0$. They will be expressed in the following lemma and its corresponding corollary.

\begin{lem}
When the serving satellite is at distance $r_0 \geq \rmin$ from the user, the number of interfering satellites, denoted by $N_I$, is a binomial random variable with success probability
\begin{align}
\label{eq:46}
\PI = \frac{\rmin-\big(r_0^2-\rmin^2\big)/(2\rEarth)}{2(\rEarth+\rmin)-\big(r_0^2-\rmin^2\big)/(2\rEarth)}.
\end{align}
\end{lem}
\begin{proof}
The expression is directly given by the ratio of the surface area where visible interfering satellites can reside (the spherical cap formed by the intersection of the user's plane and the satellites' sphere) to the total surface area of the satellites' sphere excluding the shaded cap of Fig.~2.
\end{proof}

It should be noted that there are $N/K-1$ satellites sharing the same channel with the serving satellite and potentially causing co-channel interference. The probability of having no interference ($N_I=0$) then follows from the corollary below.

\begin{cor}
When the serving satellite is at distance $r_0 \geq \rmin$ from the user, the probability of having zero co-channel interference is given by
\begin{align}
\label{eq:28}
\nonumber
P_0 &\triangleq \Prob(\NI=0) = \left(1-\PI\right)^{\frac{N}{K}-1}\\
&= \left(1-\frac{\rmin-\big(r_0^2-\rmin^2\big)/(2\rEarth)}{2(\rEarth+\rmin)-\big(r_0^2-\rmin^2\big)/(2\rEarth)}\right)^{\frac{N}{K}-1}
\end{align}
for $r_0 \leq \rmax$, and $P_0=1$ when $r_0 > \rmax$. 
\end{cor}

{Especially, the above probability turns out to be a key factor for LEO constellation analysis since, for the derivation of the following performance metrics, we can separately consider the two complementary events of having either zero interference or non-zero interference.}

\section{Coverage Probability}

In this section, we apply stochastic geometry to derive the downlink coverage probability of the LEO satellite network for a user in an arbitrary location on Earth. The performance measure of coverage probability is defined as
\begin{align}
\label{eq:3}
\Pc\left(T\right) \triangleq \Prob\left(\SINR > T\right),
\end{align}
where $T$ represents the minimum SINR required for successful data transmission. In other words, whenever the SINR of the considered user from its nearest satellite is above the threshold level $T$, it is considered to be within the coverage of the satellite communication network.

\subsection{Coverage Probability under Rayleigh Fading Channel}

In this subsection, the propagation model takes into account the large-scale attenuation with path loss exponent $\alpha$, as well as the small-scale fading. {In particular, the serving channel is assumed to follow normalized Rayleigh fading so that the corresponding channel gain is an exponential random variable with unit mean, i.e., $G_0\sim\operatorname{Exp}\left(1\right)$, while interfering channels are considered to follow general fading (but Rayleigh fading for interfering channels is also analyzed as a special case).}

The Rayleigh fading model is applicable when the line-of-sight portion of signals received at the user's place is smaller than that of non-line-of-sight components. Due to lower altitudes in LEO constellations w.r.t.\ medium Earth or geostationary orbits (MEO/GEO), a strong LoS component is less likely available for the user if the number of satellites is limited. For instance, for constellations with a fewer number of satellites such as Iridium with only 66 satellites, the probability of having even medium elevation angle is rather low \cite{40} that implies very low LoS probability. The probability of receiving LoS signals from interfering satellites is even lower due to their having smaller elevation angles. In addition, due to the shorter orbital period in LEO satellites, the channels vary rapidly with time and location which causes considerable Doppler variation at the user's location.

We express the coverage probability under the Rayleigh fading assumption for the serving channel as follows. In particular, we split the coverage probability into two terms, each of them corresponding to an important operational circumstance: The first term represents the case when there is no interference; and the second term corresponds to the case when there is at least one interfering satellite. 

\begin{theorem}
The probability of coverage for an arbitrarily located user under a Rayleigh-fading serving channel is
\begin{align}
\label{eq:12}
\Pc\left(T\right)={P_0}\,\Prob(\SNR > T) + (1 - P_0)\,\Prob(\SINR > T | \NI > 0)
\end{align}
with
\begin{align}
\label{eq:13}
\nonumber
\Prob(\SNR > T)
&= \frac{N}{2\rEarth(\rEarth+\rmin)}\int_{\rmin}^{\rmax}e^{-\frac{Tr_0^\alpha}{\ps}\sigma^2}\\
&\times \left(1-\frac{r_0^2-\rmin^2}{4\rEarth(\rEarth+\rmin)}\right)^{N-1}r_0\,dr_0
\end{align}
and
\begin{align}
\label{eq:15}
\nonumber
\Prob(\SINR > T | \NI > 0)
&= \frac{N}{2\rEarth(\rEarth+\rmin)}\int_{\rmin}^{\rmax}e^{-\frac{Tr_0^\alpha}{\ps}\sigma^2}\\
\times\mathcal{L}_{I}\left(\frac{Tr_0^\alpha}{\ps}\right)
& \left(1-\frac{r_0^2-\rmin^2}{4\rEarth(\rEarth+\rmin)}\right)^{N-1}r_0\,dr_0,
\end{align}
where $\mathcal{L}_{I}\left(s\right)$ is the Laplace transform of cumulative interference power $I$ that is expressed in Lemma~5 and $P_0$ is given in (\ref{eq:28}) by Corollary~1. 
\end{theorem}
\begin{proof}
To obtain \eqref{eq:13}, we start with the definition of coverage probability for the case of having zero interference:
\begin{align}
\label{eq:16}
\nonumber
\mathbb{E}&_{R_0}\left[\Prob\left(\SNR>T |R_0=r_0\right)\right]\\\nonumber
&=\int_{\rmin}^{\rmax}\Prob\left(\SNR>T |R_0=r_0\right)f_{R_0}\left(r_0\right)dr_0\\\nonumber
&=\frac{N}{2\rEarth(\rEarth+\rmin)}\int_{\rmin}^{\rmax}\Prob\left(\frac{\ps G_0r_0^{-\alpha}}{\sigma^2}>T \right)\\
&\hspace{80pt}\times\left(1-\frac{r_0^2-\rmin^2}{4\rEarth(\rEarth+\rmin)}\right)^{N-1}r_0\,dr_0.
\end{align}
The upper limit for the integral is due to the fact that the satellites below the user's horizon are not visible to the user. Invoking the assumption regarding the distribution of $G_0$, the expression in \eqref{eq:13} is obtained. To prove \eqref{eq:15}, we first note that there is at least one interfering satellite, so that we have
\begin{align}
\label{eq:42}
\nonumber
\mathbb{E}&_{R_0}\left[\Prob\left(\SINR>T |R_0=r_0,\NI>0\right)\right]\\\nonumber
&=\int_{\rmin}^{\rmax}\Prob\left(\SINR>T |R_0=r_0,\NI>0\right)f_{R_0}\left(r_0\right)dr_0\\\nonumber
&=\frac{N}{2\rEarth(\rEarth+\rmin)}\int_{\rmin}^{\rmax}\Prob\left(\frac{\ps G_0 r_0^{-\alpha}}{I+\sigma^2}>T \Big|\NI>0\right)\\
&\hspace{80pt}\times\left(1-\frac{r_0^2-\rmin^2}{4\rEarth(\rEarth+\rmin)}\right)^{N-1}r_0\,dr_0.
\end{align}
When $G_0\sim\operatorname{Exp}\left(1\right)$, we can express the first term in the integrand of \eqref{eq:42} as

\begin{align}
\label{eq:17}
\nonumber
\Prob&\left(G_0>\frac{Tr_0^\alpha\left(I+\sigma^2\right)}{\ps }\bigg|\NI>0\right)\\\nonumber
&=E_{I}\left[\Prob\left(G_0>\frac{Tr_0^\alpha\left(I+\sigma^2\right)}{\ps }\bigg|
\NI>0\right)\right]\\\nonumber
&=E_{I}\left[\exp\left(-\frac{Tr_0^\alpha}{\ps }\left(I+\sigma^2\right)\right)\right]\\
&=e^{-\frac{Tr_0^\alpha}{\ps }\sigma^2}E_{I}\left[e^{-\frac{Tr_0^\alpha}{\ps } I}\right]=e^{-\frac{Tr_0^\alpha}{\ps }\sigma^2}\mathcal{L}_{I}\left(\frac{Tr_0^\alpha}{\ps }\right). 
\end{align}

Substituting \eqref{eq:17} into \eqref{eq:42} completes the derivation of \eqref{eq:15}. 
\end{proof}

It is worth highlighting that, in the case of having $K=N$ orthogonal channels, the system becomes noise-limited and the coverage probability in \eqref{eq:12} will reduce to its first term.
{In the following lemma, we will obtain the Laplace function of the random variable $I$ to complete the derivation of \eqref{eq:15}.}

{
\begin{lem}
 When the serving satellite is at distance $r_0 \geq \rmin$ from the user, the Laplace transform of random variable $I$ is 
\begin{align}
\label{eq:18}
\nonumber
\mathcal{L}_{I}(s)&= \sum_{\nI=1}^{\frac{N}{K}-1} \Bigg(\binom{\frac{N}{K}-1}{\nI}\PI^{\nI}(1-\PI)^{\frac{N}{K}-1-\nI}\\
&\times {\Big(\frac{2}{\rmax^4/\rmin^2-r_0^2}\int_{r_0}^{\rmax}\mathcal{L}_{G_n}\left(s\ppi r_n^{-\alpha}\right)\,r_n\,dr_n\Big)}^{\nI}\Bigg),
\end{align}
where $\PI$ is given in \eqref{eq:46} by Lemma~4 and $\mathcal{L}_{G_n}(\cdot)$ is the Laplace transform of the random variable $G_n$. 
\end{lem}
\begin{proof}
See Appendix B.
\end{proof}
{Consequently, by only specifying $\mathcal{L}_{G_n}(\cdot)$ at the point $s\ppi r_n^{-\alpha}$, i.e., assuming some specific fading model, the Laplace transform of interference for any fading distribution can be calculated using Lemma~5. For instance, when $G_n$ is exponentially distributed, $\mathcal{L}_{G_n}(s\ppi r_n^{-\alpha})=\frac{1}{1+\ppi sr_n^{-\alpha}}$. Thus, the Laplace function of interference when interfering channels are assumed to be Rayleigh is given by the following corollary. }}
\begin{cor}
When the serving satellite is at distance $r_0 \geq \rmin$ from the user and interfering channels experience Rayleigh fading, i.e., $G_n\sim\operatorname{Exp}\left(1\right)$ for $n=1,\ldots,\NI$, the Laplace transform of random variable $I$ is 
\begin{align}
\label{eq:188}
\nonumber
\mathcal{L}_{I}(s)&=\sum_{\nI=1}^{\frac{N}{K}-1} \binom{\frac{N}{K}-1}{\nI}\PI^{\nI}(1-\PI)^{\frac{N}{K}-1-\nI}\\
&\times {\left(\frac{2}{\rmax^4/\rmin^2-r_0^2}\int_{r_0}^{\rmax}\left(\frac{r_n}{1+\ppi sr_n^{-\alpha}}\right)dr_n\right)}^{\nI}.
\end{align}
\end{cor}

{ Consequently, for Rayleigh-fading interfering channels and particular values of the path loss exponent, Lemma~5 will reduce to elementary functions, which will result in a simpler expression for \eqref{eq:15}.}
Using \cite[Eq.~3.194.5]{table}, the integral in \eqref{eq:188} can be rewritten as 
\begin{align}
\label{eq:19}
\nonumber
\int \frac{r_n}{1+\ppi sr_n^{-\alpha}}&\,dr_n\\
=-\,  \frac{r_n^2}{2}&\left({}_2F_1\left(1,\frac{2}{\alpha};1+\frac{2}{\alpha};-\, \frac{r_n^\alpha}{\ppi s}\right)-1\right),
\end{align}
where $_2F_1\left(\cdot,\cdot;\cdot;\cdot\right)$ is the Gauss's hyper-geometric function. Using the definition of the function \cite[Eq.~9.100]{table} and substituting with special arguments, it is reduced to elementary functions. If $\alpha=2$, we have \cite[Eq.~9.121.5]{table}
\begin{align}
\label{eq:20}
{}_2F_1\left(1,1;2;-\, \frac{r_n^\alpha}{\ppi s}\right)=\frac{\ppi s\,\ln{\left(1+r_n^2/(\ppi s)\right)}}{r_n^2}
\end{align}
so that \eqref{eq:188} can be rewritten as
\begin{align}
\label{eq:43}
\nonumber
\mathcal{L}_{I}(s)&=\sum_{\nI=1}^{\frac{N}{K}-1} \binom{\frac{N}{K}-1}{\nI}\PI^{\nI}(1-\PI)^{\frac{N}{K}-1-\nI}\\
&\times\frac{\rmax^2-r_0^2}{\rmax^4/\rmin^2-r_0^2}~{\left(\frac{\ppi s}{\rmax^2-r_0^2}~{\ln\left(\frac{\ppi s+r_0^2}{\ppi s+\rmax^2}\right)}+1\right)}^{\nI}.
\end{align}
On the other hand, if we assume $\alpha=4$, the Gauss hyper-geometric function can be written as \cite[Eq.~9.121.27]{table}
\begin{align}
\label{eq:21}
{}_2F_1\left(1,\frac{1}{2};\frac{3}{2};-\, \frac{r_n^4}{\ppi s}\right)=\frac{\sqrt{\ppi s}\,\arctan{\left(r_n^2/\sqrt{\ppi s}\right)}}{r_n^2},
\end{align}
by which \eqref{eq:188} is simplified into a closed-form expression as
\begin{align}
\label{eq:44}
\nonumber
\mathcal{L}_{I}(s)
&=\sum_{\nI=1}^{\frac{N}{K}-1} \binom{\frac{N}{K}-1}{\nI}\PI^{\nI}(1-\PI)^{\frac{N}{K}-1-\nI}\\
\times&\frac{\rmax^2-r_0^2}{\rmax^4/\rmin^2-r_0^2}~{\left(\frac{\sqrt{\ppi s}}{\rmax^2-r_0^2}~{\arctan\left(\frac{\sqrt{\ppi s}\left(r_0^2-\rmax^2\right)}{{\ppi s}+\rmax^2r_0^2}\right)}+1\right)}^{\nI}.
\end{align}

\subsection{Coverage Probability under Non-fading Channels}

In this subsection, we derive the coverage probability for non-fading propagation environments. { This propagation model is applicable (at least as an accurate approximation) when the number of satellites in a constellation is large enough, so that it is likely to have multiple satellites in LoS propagation range. Consequently, the serving satellite is likely high above the user and potential multi-path fading components are weak compared to the direct propagation path.} Moreover, the LoS probability for interfering satellites increases as we decrease the number of frequency channels. 

Similar to Theorem~1, the coverage probability is split into two terms. The first one is for the case of having no interference and, in the second term, we assume that there is at least one interfering satellite. As will be seen shortly, other than representing two important communication scenarios, such division is needed to make our approach tractable. For the non-fading serving channel, we denote $G_0=1$. { Like in the previous section, any general fading statistics can be assumed for interfering channels, for which non-fading interference is considered as a special case towards the end of the section.}

\begin{theorem}
The coverage probability of a user with a non-fading serving channel is
\begin{align}
\label{eq:29}
\Pc(T)={P_0}\,\Prob(\SNR>T) + (1-P_0)\,\Prob(\SINR>T|\NI>0),
\end{align}
with
\begin{equation}
\label{eq:399}
\Prob(\SNR>T)=F_{R_0}\left(\left(\frac{\ps}{T\sigma^2}\right)^{\frac{1}{\alpha}}\right)
\end{equation}
and
\begin{align}
\label{eq:39}
\nonumber
\Prob(&\SINR>T|\NI>0)
=\frac{N}{4\pi \rEarth(\rEarth+\rmin)}\int_{\rmin}^{\rmax}\int_{-\infty}^{\infty}\mathcal{L}_{I}(j\omega)\\
&\times\left(\frac{e^{j\left({\frac{\ps }{Tr_0^\alpha}-\sigma^2}\right)\omega}-1}{j\omega} 
\right)\left(1-\frac{r_0^2-\rmin^2}{4\rEarth(\rEarth+\rmin)}\right)^{N-1}r_0\,dr_0d\omega.
\end{align}
\end{theorem}
\begin{proof}
The derivation of \eqref{eq:399} is straightforward by substituting $G_0 \to 1$ in \eqref{eq:41} and using the definition of a CDF.
Excluding zero interference from $I$ satisfies the sufficient conditions, given in \cite[Prop.~A.2]{37}, for its PDF $f_{I}(\cdot)$ to exist (especially, $I$ is continuous). Substituting $G_0\to1$ in the first term of the integral in \eqref{eq:42}, we have
\begin{align}
\label{eq:34}
 \Prob&\left(Tr_0^\alpha(\sigma^2+I)<\ps \big|\NI>0\right)= \Prob\left(I<\frac{\ps }{Tr_0^\alpha}-\sigma^2\Big|\NI>0\right)\\
 \label{eq:34-2}
& =\int_{0}^{\frac{\ps }{Tr_0^\alpha}-\sigma^2}f_{I}(i)di\\
\label{eq:34-3}
&=\frac{1}{2\pi}\int_{-\infty}^{\infty}\mathcal{L}_{I}(s)|_{s=j\omega}\frac{e^{j\left({\frac{\ps }{Tr_0^\alpha}-\sigma^2}\right)\omega}-1}{j\omega}d\omega.
\end{align}
The latter equality follows immediately from the Parseval--Plancherel property and from the fact that the Fourier transform of the square integrable function $\mathds{1} \big(0\leq t \leq \frac{\ps }{Tr_0^\alpha}-\sigma^2\big)$ is
$\big(1-e^{-j({\frac{\ps}{Tr_0^\alpha}-\sigma^2})\omega}\big)/(j\omega)$ according to \cite{29}.
\end{proof}

The Laplace function $\mathcal{L}_{I}(s)$ for general fading can be obtained using Lemma~5 given in the previous section. The following corollary presents the Laplace function when interfering signals experience non-fading propagation as well.

{
\begin{cor}
When the serving satellite is at distance $r_0 \geq \rmin$ from the user, the Laplace transform of $I$ for non-fading channels is
\begin{align}
\label{eq:36}
\nonumber
&\mathcal{L}_{I}(s)=\sum_{\nI=1}^{\frac{N}{K}-1} \binom{\frac{N}{K}-1}{\nI}\PI^{\nI}(1-\PI)^{\frac{N}{K}-1-\nI}\\
&\times \left(\frac{2(s\ppi )^{2/\alpha}}{\alpha\left(\rmax^4/\rmin^2-r_0^2\right)}\left[\Gamma(-2/\alpha,s\ppi \rmax^{-\alpha})-\Gamma(-2/\alpha,s\ppi r_0^{-\alpha})\right]\right)^{\nI},
\end{align}
where $\Gamma(a,x)=\int_x^\infty y^{a-1}e^{-y}dy$ denotes the upper incomplete gamma function. 
\end{cor}
\begin{proof}
The first term of the integrand in Lemma~5 will be reduced to an exponential function. As a result, the integral can be expressed in form of the upper incomplete gamma function defined as $\Gamma(a,x)=\int_x^\infty y^{a-1}e^{-y}dy$ by changing integration variable as $s\ppi r_n^{-\alpha}\to y$.
\end{proof}}

\section{Average Achievable Rate}

In this section, we focus on the average achievable data rate. The technical tools and expressions presented in Sections~II and III are reused also herein.
The average achievable rate (in bit/s/Hz) is defined as
\begin{align}
\label{eq:45}
\bar{C} \triangleq\frac{1}{K}~\mathbb{E}\left[\log_2\left(1+\SINR\right)\right],
\end{align}
which states the ergodic capacity from the Shannon--Hartley theorem over a fading communication link normalized to the bandwidth of $1/K$ [Hz].

\subsection{Average Achievable Rate under Rayleigh Fading Channel}

We can calculate the expression for the average rate of an arbitrary user under the assumption of a Rayleigh-fading serving channel as follows. It is worth noting that the average is taken over both the spatial BPP and fading distribution.

\begin{theorem}
The downlink average rate (in bits/s/Hz) of an arbitrarily located user and its serving satellite under Rayleigh fading assumption, i.e., $G_0\sim\operatorname{Exp\,(1)}$, is
\begin{align}
\label{eq:22}
\nonumber
\bar{C} &\triangleq\frac{P_0}{K}\,\mathbb{E}\left[\log_2\left(1+\SNR\right)\right]\\
&\hspace{10pt}+\frac{1-P_0}{K}\,\mathbb{E}\left[\log_2\left(1+\SINR\right)|\NI > 0\right],
\end{align}
with
\begin{align}
\label{eq:37}
\nonumber
\mathbb{E}&\left[\log_2\left(1+\SNR\right)\right]\\\nonumber
&=\frac{N}{2 \ln(2)\rEarth(\rEarth+\rmin)}\int_{\rmin}^{\rmax}\int_{t>0}e^{-\frac{r_0^\alpha}{\ps }\sigma^2 \left(e^t-1\right)}\\
&\hspace{10pt}\times\left(1-\frac{r_0^2-\rmin^2}{4\rEarth(\rEarth+\rmin)}\right)^{N-1}r_0\,dtdr_0
\end{align}
and
\begin{align}
\label{eq:33}
\nonumber
\mathbb{E}&\left[\log_2\left(1+\SINR\right)|\NI>0\right]\\\nonumber
&=\frac{N}{2 \ln(2) \rEarth(\rEarth+\rmin)}\int_{\rmin}^{\rmax}\int_{t>0}e^{-\frac{r_0^\alpha}{\ps }\sigma^2 \left(e^t-1\right)}\\
&\hspace{10pt}\times\mathcal{L}_{I}\left(\frac{r_0^\alpha}{\ps } \left(e^t-1\right)\right)\left(1-\frac{r_0^2-\rmin^2}{4\rEarth(\rEarth+\rmin)}\right)^{N-1}r_0\,dtdr_0.
\end{align}
\end{theorem}
\begin{proof}
See Appendix~C.
\end{proof}

{The above applies to interfering channels with any fading statistics. However, analogous to the results in Section~III, by considering $G_n\sim\operatorname{Exp}\left(1\right)$ for some specific $\alpha$ values, i.e., $\alpha=2$ and $\alpha=4$}, and substituting \eqref{eq:43} or \eqref{eq:44} into \eqref{eq:33} will lead to more simplified expressions for Theorem~3. Moreover, if we allocate $K=N$ orthogonal channels to the system, the second term in \eqref{eq:22} will be eliminated and the network's performance will become noise-limited.

\subsection{Average Achievable Rate under Non-fading Channels}

In this subsection, we derive the average data rate for non-fading serving channels and any fading statistics for interference. The rate expression is again split into two terms {to consider zero and non-zero interference conditions separately.}

\begin{theorem}
The downlink average rate of an arbitrary located mobile user and its serving satellite for a non-fading channel is 
\begin{align}
\label{eq:222}
\nonumber
\bar{C} &\triangleq\frac{P_0}{K}\,\mathbb{E}\left[\log_2\left(1+\SNR\right)\right]\\
&\hspace{10pt}+\frac{(1-P_0)}{K}\,\mathbb{E}\left[\log_2\left(1+\SINR\right)|\NI > 0\right],
\end{align}
with
\begin{equation}
\label{eq:31}
\mathbb{E}\left[\log_2\left(1+\SNR\right)\right]=\frac{1}{\ln(2)}\int_{t>0} F_{R_0}\left(\left(\frac{\ps }{\left(e^t-1\right)\sigma^2}\right)^{\frac{1}{\alpha}}\right) dt
\end{equation}
and 
\begin{align}
\label{eq:2222}
\nonumber
&\mathbb{E}\left[\log_2\left(1+\SINR|\NI>0\right)\right]\\\nonumber
&=\frac{N}{4\pi \ln(2)\rEarth(\rEarth+\rmin)}\int_{\rmin}^{\rmax}\int_{t>0}\int_{-\infty}^{\infty}\mathcal{L}_{I}(j\omega)\\
&\hspace{1pt}\times\left(\frac{e^{j\left({\frac{\ps }{r_0^\alpha\left(e^t-1\right)}-\sigma^2}\right)\omega}-1}{j\omega} 
\right)\left(1-\frac{r_0^2-\rmin^2}{4\rEarth(\rEarth+\rmin)}\right)^{N-1}r_0 \, d\omega \, dt \, dr_0.
\end{align}
\end{theorem}
\begin{proof}
See Appendix~D.
\end{proof}
{In case of non-fading interfering channels, the Laplace transform in \eqref{eq:2222} can be calculated from Corollary~3.} 
\begin{figure}
    \includegraphics[trim = 5mm 0mm 5.5cm 17cm, clip,width=\textwidth]{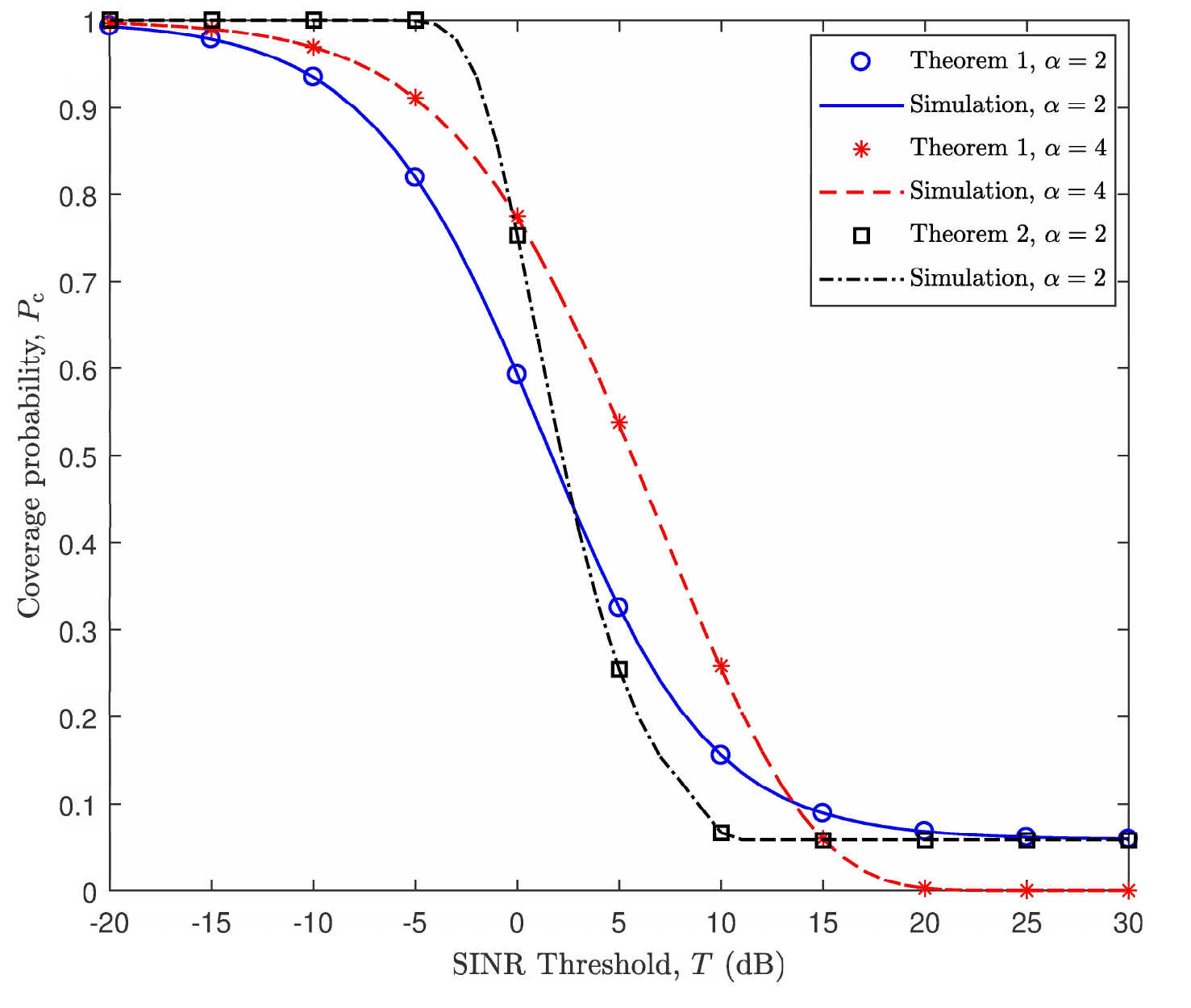}
    \caption{Verification of coverage probability expressions with simulations for $N=720$ satellites. The path loss exponent ($\alpha$) is set to 2 and 4.}
    \label{fig:3}
\end{figure}
\section{Numerical Results}

In this section, we verify our expressions for coverage probability and average data rate by comparing them with the results of Monte Carlo simulations. { We obtain a new parameter---effective number of satellites ($\Neff$)---to compensate for the effect of uneven distribution of satellites along different latitudes.} We then illustrate the two performance metrics in terms of different network design parameters such as the number of frequency channels and the satellite altitude.

For numerical verification, we calculate the simulated coverage probability in Monte Carlo manner as follows. First, we randomly place $N$ satellites with uniform distribution on a sphere centered at Earth's center and distanced $\rmin$ from its surface. At the same time, we also model the channels between the satellites and the user by exponential random variables (to model Rayleigh fading) with unit mean value while, to model a non-fading environment, we set the channel gains to one. Then we calculate the SINR at the receiver and compare it with a pre-defined threshold value to evaluate coverage probability. We repeat this experiment for a large number of realizations to obtain the averaged performance metrics. Finally, we compare these numerical results with the analytical expressions in Theorems 1--4 to confirm our analysis.

{For producing the numerical results, the transmitted power from all satellites is set to $\ps=\ppi=10$~W. The assumption corresponds to the case where all satellites are equipped with omni-directional antennas. The theory is applicable also to other antenna patterns by adjusting different power levels for the serving and interfering satellites accordingly to take into account that the serving satellite's transmission would be directed to the user while all interfering transmissions are sent to other, arbitrary and independent, directions. The noise power is assumed to be $\sigma^2=-98$~dBm.} The number of channels is considered to be $K=20$ unless otherwise stated. 

{\subsection{Corroboration of Theorems by Simulations}

In this subsection, we validate the coverage probability and the data rate derived in Theorems 1--4 through Monte Carlo simulations. We considered static propagation or Rayleigh fading for serving channels and general fading for interference.}

\begin{figure}

         \includegraphics[trim = 5mm 0mm 5.5cm 17cm, clip,width=\textwidth]{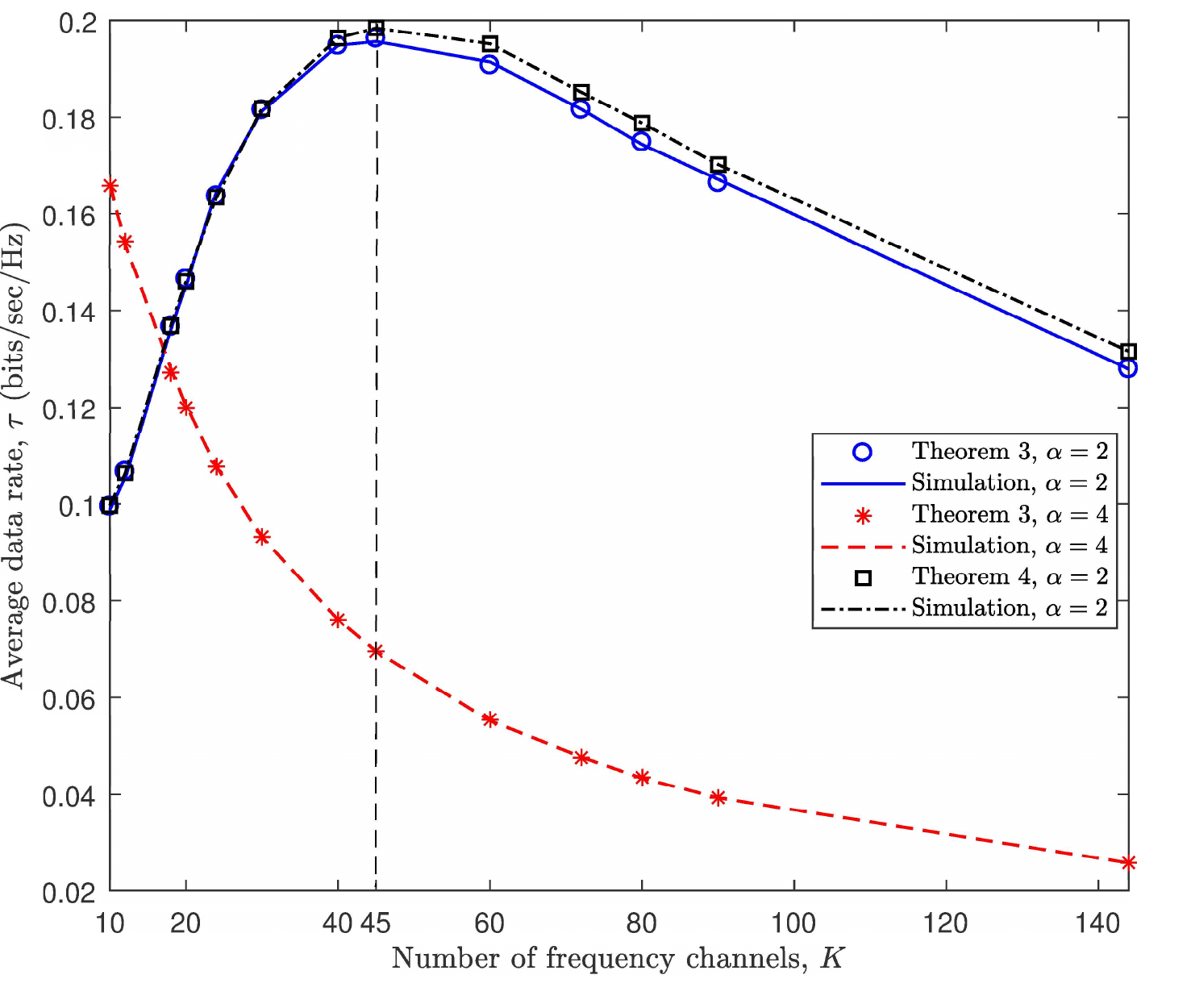}
        \caption{Effect of the number of frequency bands ($K$) on average achievable rate with Rayleigh and non-fading channels.}
        \label{fig:4}
\end{figure}

Figure~\ref{fig:3} demonstrates the coverage probability for different threshold values assuming $\alpha$ is 2 and 4. The parameters $N$ and $\rmin$ are set to $720$ and $1200$ km, respectively. We find that our theoretical results are perfectly in line with the simulations, which confirms the correctness of the derivations. {The transition of coverage probability from one to zero occurs variously depending on path loss exponent and the type of fading. For $\alpha=2$, the coverage probability saturates when $T>20$~dB which corresponds to the case when the number of interfering satellites is considerable. As a result, the $\SINR$ and, consequently, the second term in Theorem~1 approach zero. For larger path loss exponents, the effect of interference becomes insignificant and the transition from one to zero includes no saturation mode since the difference between $\SNR$ and $\SINR$ becomes less dominant. All curves tend to zero for thresholds greater than 80~dB although not shown in the figure.}

{Figure~\ref{fig:4} shows average achievable data rate versus the number of frequency bands for Rayleigh fading and non-fading environments. The behavior of the curves can be justified according to two contradictory effects of increasing the number of orthogonal channels on data rate. Allocating more orthogonal channels improves data rate by mitigating the interference but, at the same time, degrades it by making only a portion of the whole band available for each group of satellites. 
For $\alpha=2$, the former effect dominates for $K<45$ as by increasing $K$, data rate increases accordingly, while for $K>45$, the latter effects overcomes the interference elimination effect and data rate starts falling.

The effect of increasing the number of channels on the reduced available frequency band dominates the SINR enhancement for $\alpha=4$ since the largest average data rate for $\alpha=4$ corresponds to the lowest number of frequency channels. This is due to the fact that the larger path loss exponents make the interfering satellites less effective since they transmit from a farther distance. As it can be observed in the figure, for AWGN channel and $\alpha=2$, the behaviour of the data rate is slightly affected by the fading model.}

\begin{figure}
    \includegraphics[trim = 5mm 0mm 5.5cm 17cm, clip,width=\textwidth]{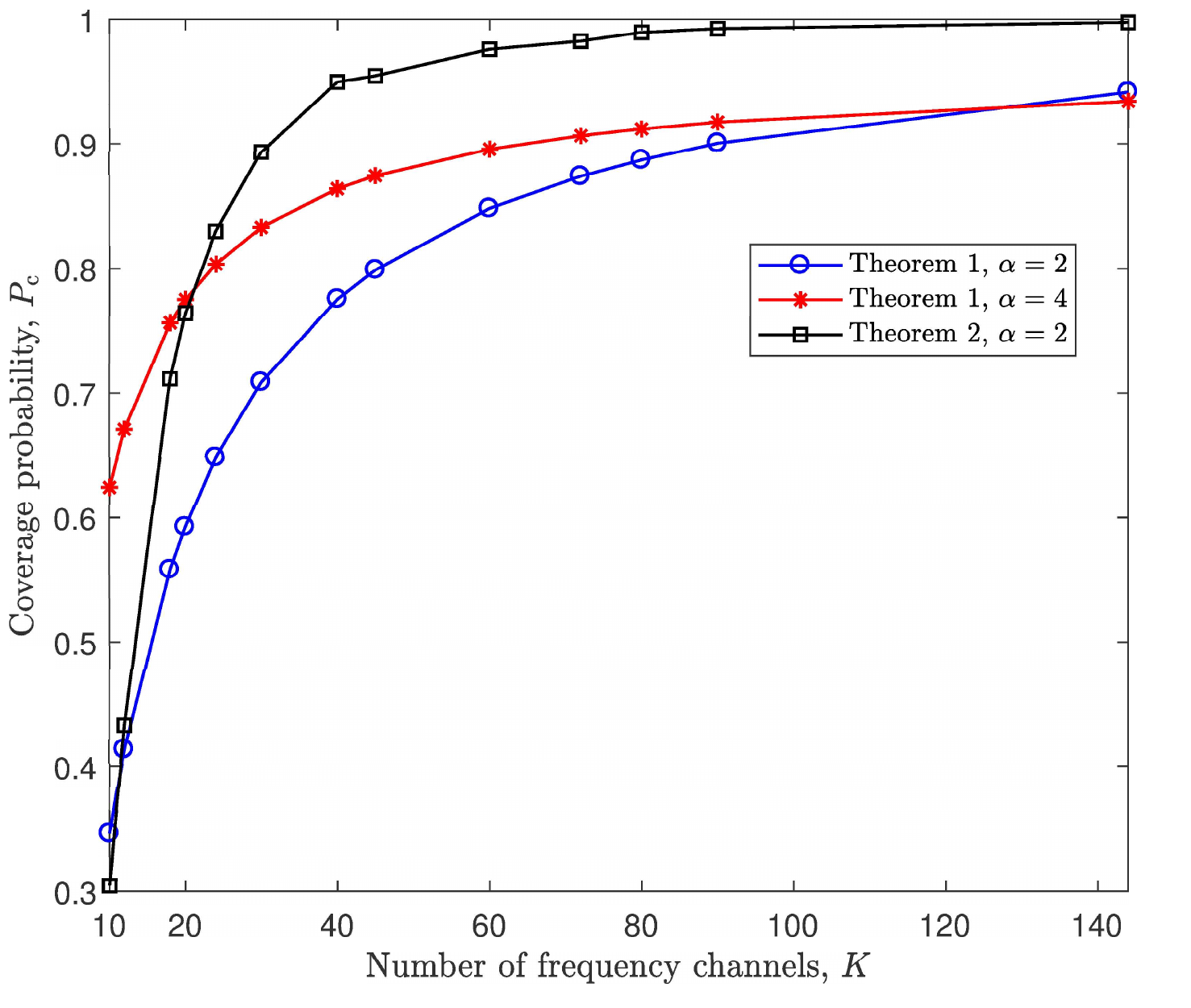}
    \caption{Effect of frequency reuse on coverage probability for different propagation environments and path loss exponents.}
    \label{fig:8}
\end{figure}

The effect of frequency reuse on coverage probability can be determined from Fig.~\ref{fig:8}. The channel allocation policy is such that the channel corresponding to the nearest satellite to the user will be selected. As can be observed, the coverage probability rises with increasing the number of frequency bands due to interference mitigation, and all three curves present the same behaviour. {However, the effect of the parameter $K$ is more significant for the smaller path loss exponent as it affects the interference level more significantly. 
The coverage performance is superior for non-fading environment due to better quality of the serving signal. It can be interpreted from Figs.~\ref{fig:4} and \ref{fig:8} that the coverage probability is an increasing function of the number of channels, i.e., the higher number of frequency channels will result in better coverage probability, while there is an optimum number of channels that maximizes the data rate. Therefore, the number of frequency channels should be compromised according to the performance demands of the intended constellation.}

{\subsection{Effective Number of Satellites}}

In this subsection, we provide analysis of deterministic Walker satellite constellations with different inclination angles which are expected to provide global Internet broadband services to users by 720 LEO satellites. In Fig.~\ref{fig:5}, the coverage probability of the Walker constellation, with inclination angles, $\phii$, of $90^\circ$ (polar), $70^\circ$ and $40^\circ$ are plotted for comparison to a random constellation. The user's latitude, denoted by $\phiu$, is assumed to be $30^\circ$.
The slight difference of coverage probability in random and Walker constellations is due to the fact that the latitudal density of satellites in deterministic constellations is not uniform, as there are more satellites around their inclination angle (i.e., the latitude limits) and less in the equatorial regions. Consequently, the coverage probability changes as we move from the equator to the latitude limits. 

\begin{figure}
    \includegraphics[trim = 5mm 0mm 5.4cm 17cm, clip,width=\textwidth]{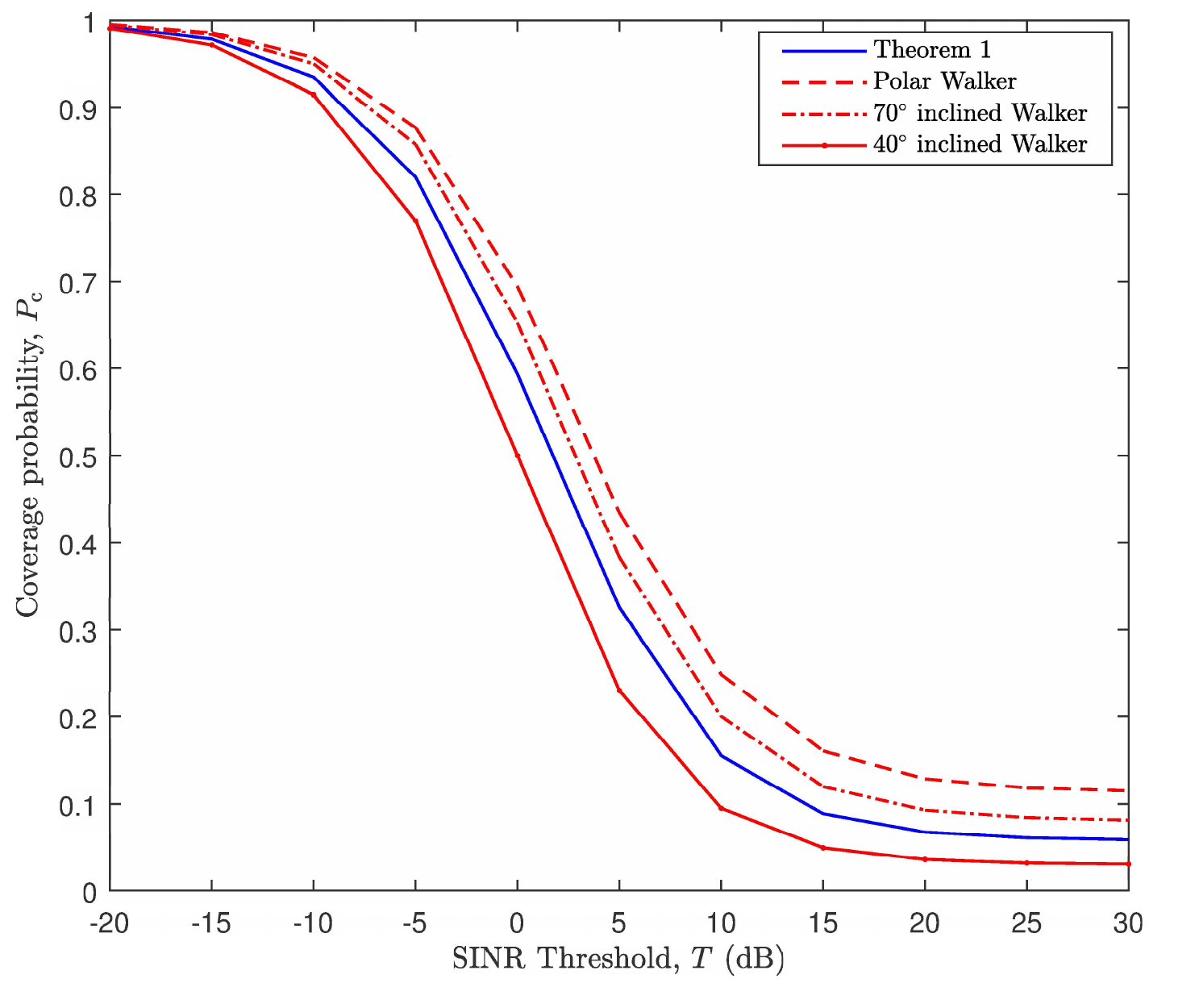}
    \caption{ Comparison of Theorem 1 with actual polar and inclined constellations. The user's latitude is set to $30^\circ$. }
    \label{fig:5}
\end{figure}

\begin{figure}
    \includegraphics[trim = 5mm 0mm 5.4cm 17cm, clip,width=\textwidth]{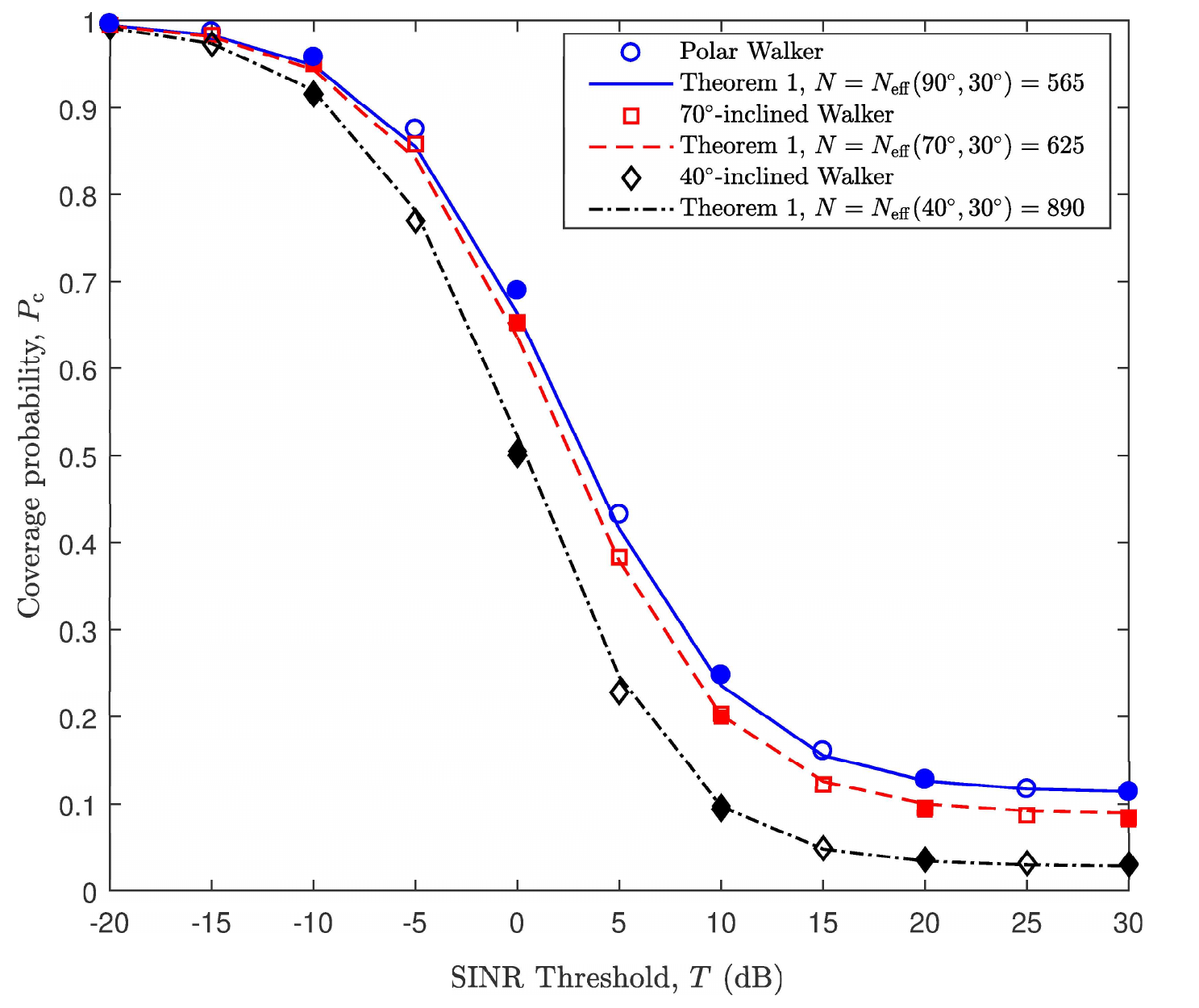}
    \caption{{Effect of using $\Neff$ in Theorem 1 on providing a better matching between the random and practical constellations. Filled points have been used in mean absolute error minimization process in order to obtain $\Neff$. }}
    \label{fig:7}
\end{figure}

{In order to compensate for uneven latitudal density, we introduce a new parameter, entitled as the effective number of satellites, and denote it by $\Neff(\phii,\phiu)$, which is the number of satellites for a random uniform constellation that corresponds to the same coverage probability as for a practical non-uniform constellation with inclination angle $\phii$ and at the user latitude $\phiu$. This parameter is approximated by mean absolute error minimization between the coverage probability given in Theorem 1 and the results from practical Walker constellation for a few threshold values on different latitudes. The approximated value for $\Neff$ is then refined to give the best matching between random and practical constellation based on the target performance metric. Although we used some of the values from simulated results to obtain $\Neff$, the same $\Neff$ can be used for other system parameters, i.e., path loss exponent and/or SINR threshold values. Likewise, the acquired $\Neff$ can be also used for a constellation with different total number of satellites by linearly scaling it accordingly.

\begin{figure}
    \includegraphics[trim = 4.5mm 0mm 4.6cm 17cm, clip,width=\textwidth]{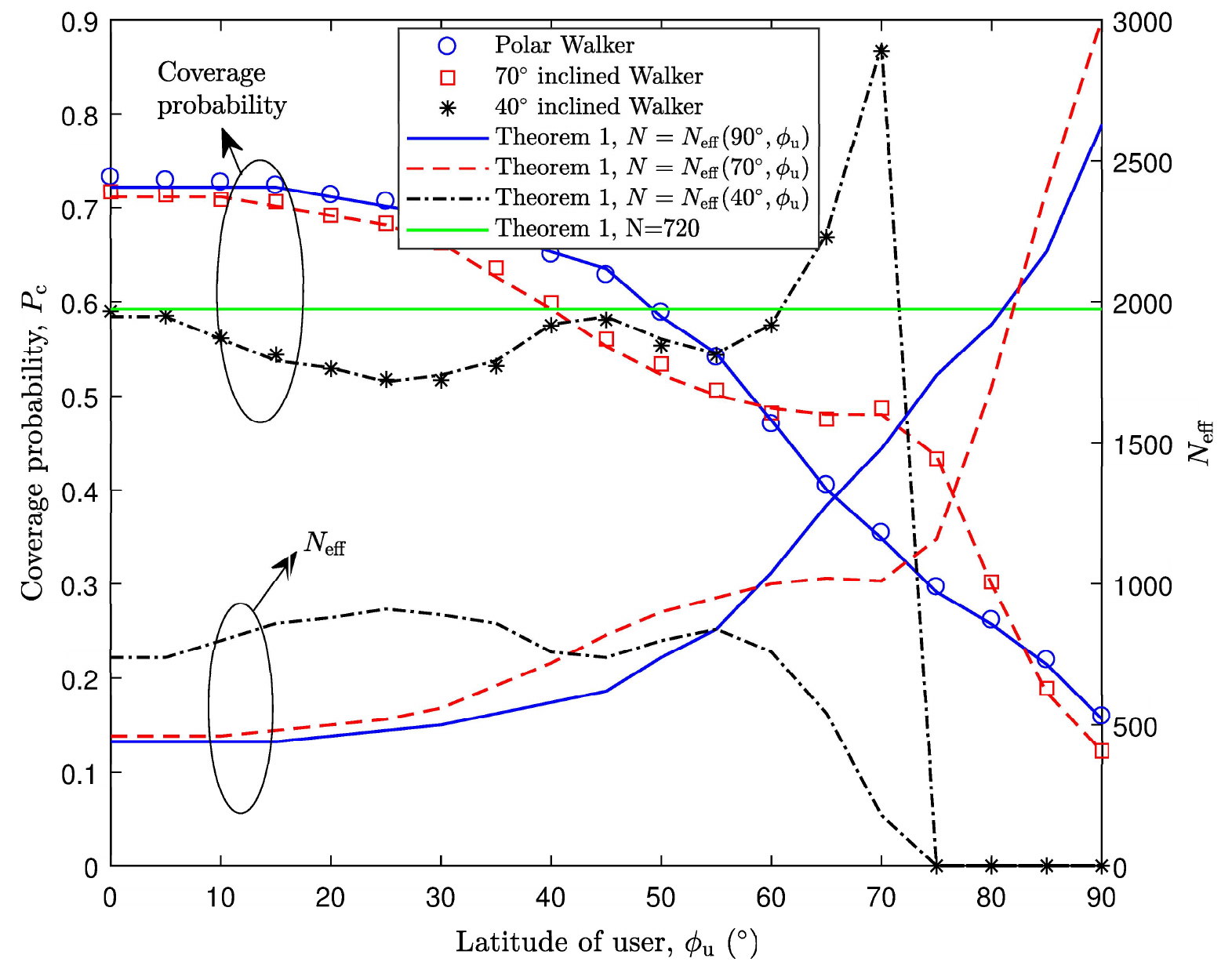}
    \caption{Effect of latitude on coverage probability. The markers for coverage probability are plotted by simulation of a Walker constellation for $90^\circ$, $70^\circ$ and $40^\circ$ inclined orbits. The right axis depicts $\Neff(\phii,\phiu)$ for different latitudes. The lines are plotted by applying $\Neff(\phii,\phiu)$ in Theorem 1. }
    \label{fig:6}
\end{figure}

\begin{figure}
         \includegraphics[trim = 4.5mm 0mm 5cm 17cm, clip,width=\textwidth]{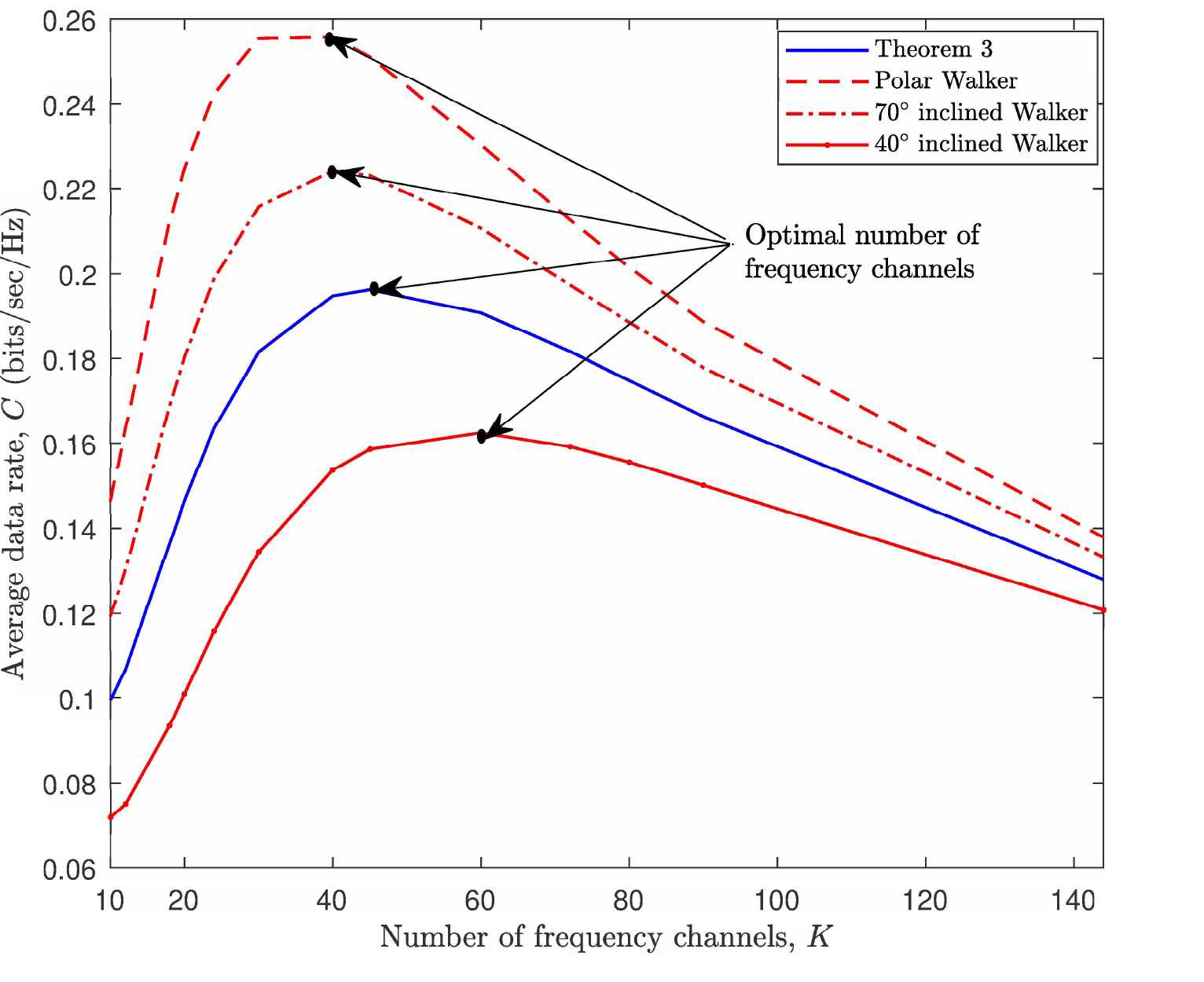}
        \caption{{Comparison of Theorem 3 with actual polar and inclined constellations. The user's latitude is set to $30^\circ$.}}
        \label{fig:9}
\end{figure}

In Fig.~\ref{fig:7}, we eliminated the deviation between Walker and random constellations (cf.\ Fig.~\ref{fig:5}) by compensating the non-uniform latitudal density. The solid markers on the curves were chosen from practical constellation in order to minimize the mean absolute error in coverage probability of a random constellation. As it can be seen, although only a few points are employed for fitting, the selected $\Neff$ also holds for other values that had no contribution on the minimization process. In other words, applying $\Neff$ in theorems provided in this paper will lead to more precise results for many system parameters and performance metrics by taking into account only a few input data from the actual network.}

In Fig.~\ref{fig:6}, the effect of latitude of user on coverage probability is shown for three practical constellations.  
In this figure, we have two groups of data that demonstrate the coverage probability. The first group (depicted by markers) is obtained by simulation of a Walker constellation for polar, $70^\circ$ and $40^\circ$ inclined orbits. For a polar constellation, the number of satellites in the view increases monotonically from the equator to poles. As a result, the coverage probability declines as the user moves from the equator to poles due to intensifying interference. On the other hand, for the other two inclination angles, the number of satellites in the view varies non-uniformly due to the network geometry which results in coverage alternation for different latitudes. {The zero coverage probability for user latitudes greater than $75^\circ$ in $40^\circ$-inclined Walker results from the fact that there are no visible satellites for those latitudes.}

The second group of data (depicted by lines) in Fig.~\ref{fig:6} is generated using Theorem 1 with $\Neff(\phii,\phiu)$ satellites. As can be seen in the figure, coverage probability in Theorem~1 for $\Neff(\phii,\phiu)$ results in the same coverage from Walker constellation for the given $\phiu$. Therefore, using $\Neff(\phii,\phiu)$ in Theorem 1, the effect of uneven distribution of satellites along different latitudes is compensated and the coverage probability can be obtained for any actual constellation. Most importantly, we can conclude that the approximation error in modeling deterministic constellations as random ones is actually rather minimal when the varying satellite density at different latitudes is taken into account.

\begin{figure}
         \includegraphics[trim = 4.5mm 0mm 4.6cm 17cm, clip,width=\textwidth]{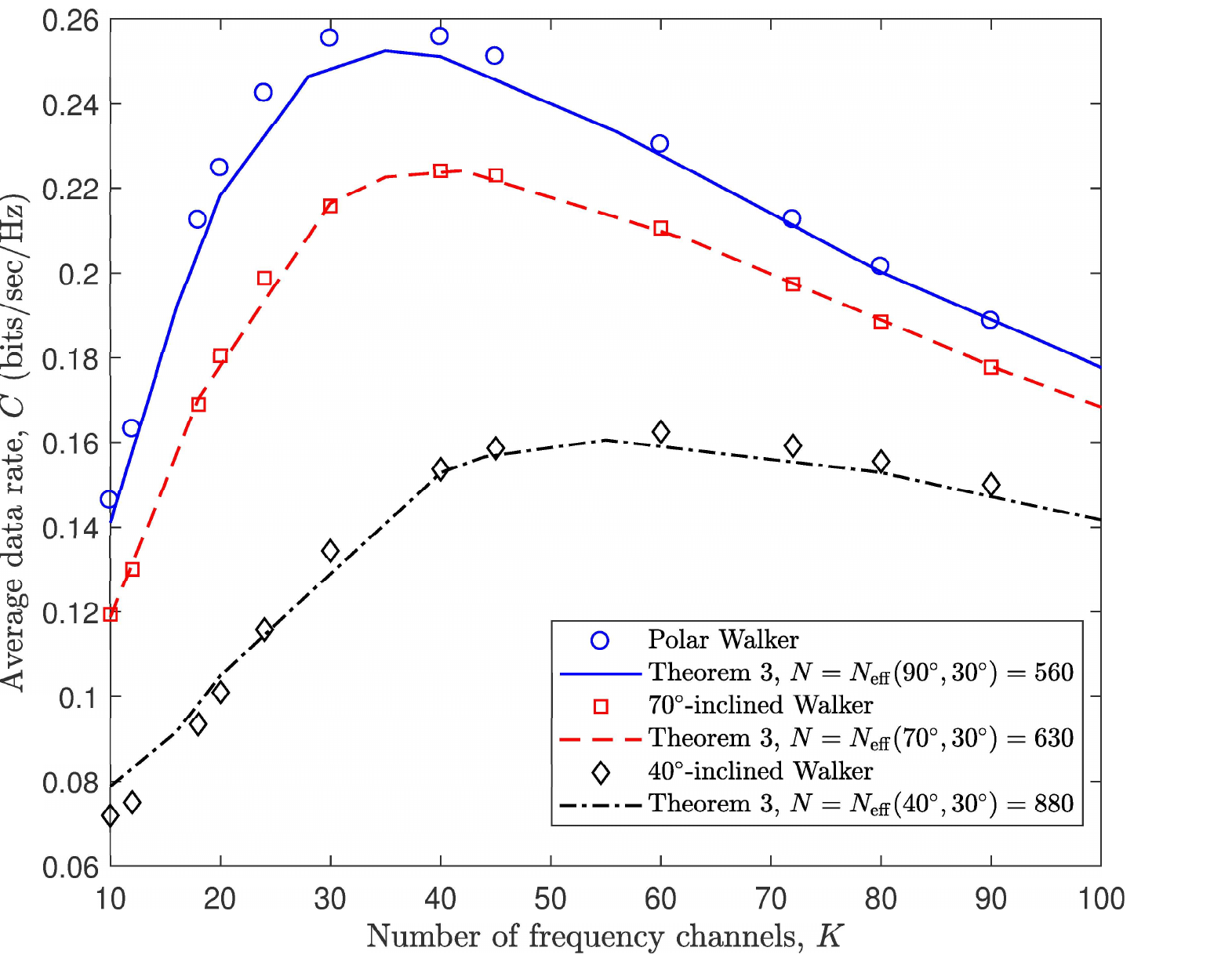}
        \caption{{Effect of using $\Neff$ in Theorem 3 on providing a better matching between the random and practical constellations.}}
        \label{fig:10}
\end{figure}

The data rate of a Walker satellite constellation for different inclination angles to compare with a random constellation (viz.\ Theorem 3) is depicted in Fig.~\ref{fig:9}. As mentioned for Fig.~\ref{fig:5}, the difference between the random and Walker constellation is mainly due to the non-uniform density of satellites in Walker constellation. { The optimal number of frequency channels that can be observed in this figure for each of the constellations increases by decreasing the inclination angle. Figure~\ref{fig:10} depicts the effect of applying $\Neff$ to Theorem 3 in eliminating the observed deviation in Fig.~\ref{fig:9}. The selected $\Neff$ values are slightly different from those in Fig.~\ref{fig:7} since they are refined to create the best possible matching between the data rate of the random and practical constellation.}

Figure~\ref{fig:11} is a counterpart for Fig.~\ref{fig:6} to illustrate the effect of latitude of the user on achievable data rate and $\Neff(\phii,\phiu)$. In this figure, $\Neff(\phii,\phiu)$ is the number of satellites in a random constellation that corresponds to the same data rate as a Walker constellation for the given $\phiu$. Due to different values of $\Neff(\phii,\phiu)$ in Figs.~\ref{fig:6} and~\ref{fig:11}, we can realize that this value is not only a function of geometrical characteristics of the constellation but also depends slightly on the performance metric. However, the compensation of latitudal density is still applicable to varying system parameters that are different from those used for fitting theoretical expressions to simulations.
\vspace{20pt}

\subsection{The Effect of Satellite Altitude}

Coverage probability for different altitudes is then depicted in Fig.~\ref{fig:12}. The figure is plotted using $\Neff(\phii,\phiu)$ for Theorem~1. {In this figure, we assumed that the total number of satellites varies with altitude so that $\Neff$ remains constant irrespective of altitude. This balances between a practical constellation design in which the number of satellites (when aiming at covering economically the entire globe) typically decreases with increasing altitude and the direct relationship between altitude and the effective number of satellites.}
It can be observed from Fig.~\ref{fig:12} that, for practical LEO altitudes, the coverage probability declines when altitude increases. Within Earth's atmosphere, there is an impractical optimal altitude, which moves even lower as the number of satellites increases.

Data rate for different altitudes is plotted in Fig.~\ref{fig:13}. The figure is created using Theorem~3 with $\Neff(\phii,\phiu)$ that corresponds to a specific latitude $\phiu$ for Walker constellations (polar and inclined). {The same assumption as for Fig.~\ref{fig:12} is also made regarding the invariance of $\Neff$ w.r.t.\ altitude.} This serves to show that, having $\Neff(\phii,\phiu)$ at hand, we can use the expressions in Theorems 3 or 4 to analyze the data rate for any given satellite constellation and the approximation error in modeling a deterministic constellation as a binomial point process is insignificant. The same as for Fig.~\ref{fig:12}, it can be seen that the peak value of data rate does not exist within the LEO practical altitude range outside atmospheric drag.

\begin{figure}
         \includegraphics[trim = 5mm 1mm 4.3cm 16.3cm, clip,width=\textwidth]{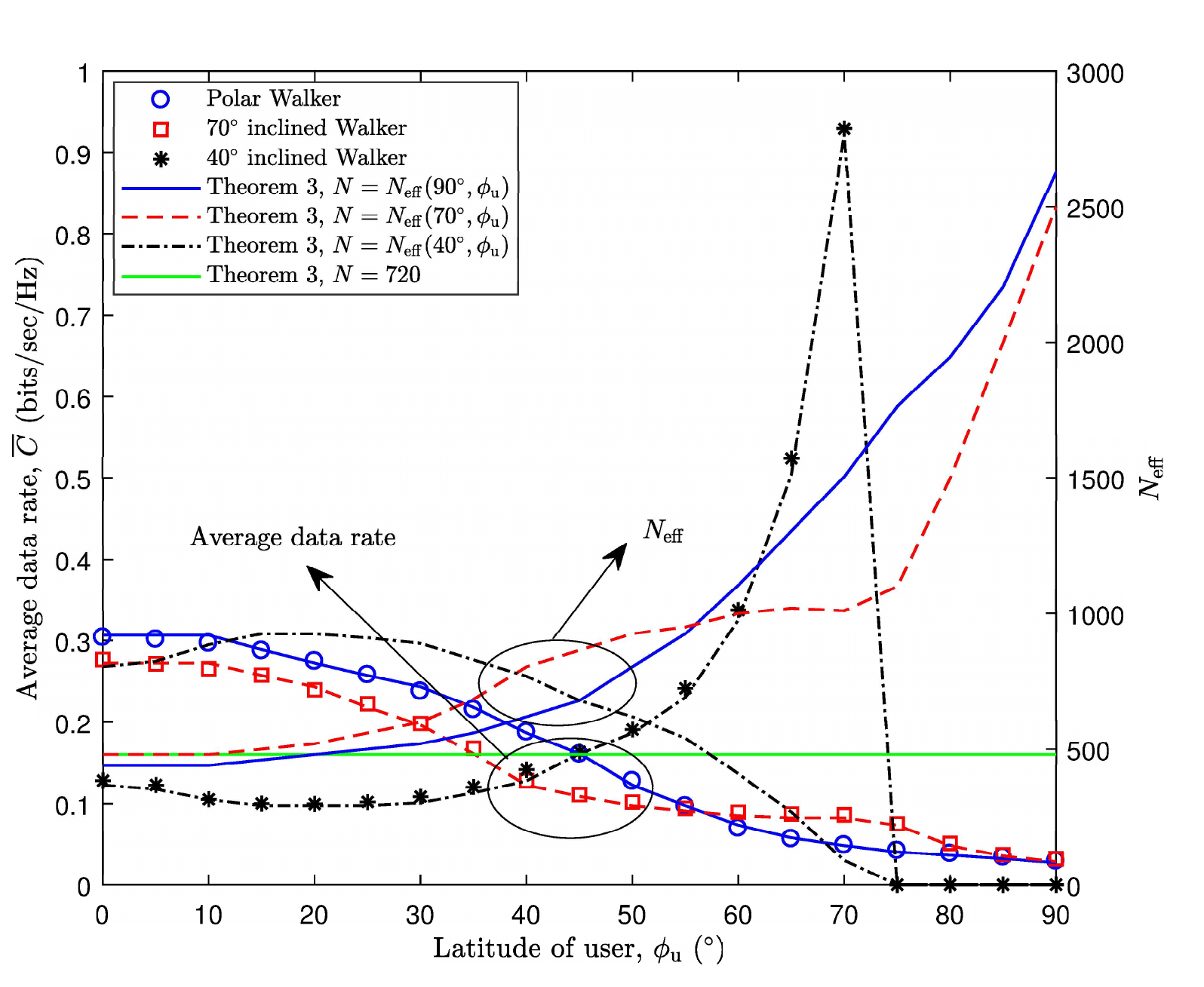}
        \caption{Effect of latitude on average data rate and $\Neff(\phii,\phiu)$. The markers for data rate are plotted by simulation of a Walker constellation for $90^\circ$, $70^\circ$ and $40^\circ$ inclined orbits. The lines are plotted by applying $\Neff(\phii,\phiu)$ in Theorem 3. }
        \label{fig:11}
\end{figure}

\section{Conclusions}
 
In this paper, we presented a tractable approach for downlink coverage and rate analysis of low Earth orbit satellite networks. Using the concept of stochastic geometry, we modeled the satellite network as a binomial point process. We then applied this model to obtain exact expressions for coverage probability and data rate of an arbitrary user in terms of network parameters and the Laplace transform of interference power. The performance metrics of the random and real constellation match with each other almost perfectly; there is only a slight deviation between them which can be compensated by taking into account the effect of uneven satellite distribution along different latitudes. A frequency reuse approach was also applied to make the network spectral efficiency suitable for commercial operation. Its effect on data rate is more ambivalent than on coverage probability. The proposed framework in this paper paves the way for accurate analysis and design of the future dense satellite networks. The scalable results presented here are not only applicable to a single constellation but also to massive satellite networks that merge several different constellations.

\begin{figure}
    \includegraphics[trim = 5mm 1mm 4.6cm 16.6cm, clip,width=\textwidth]{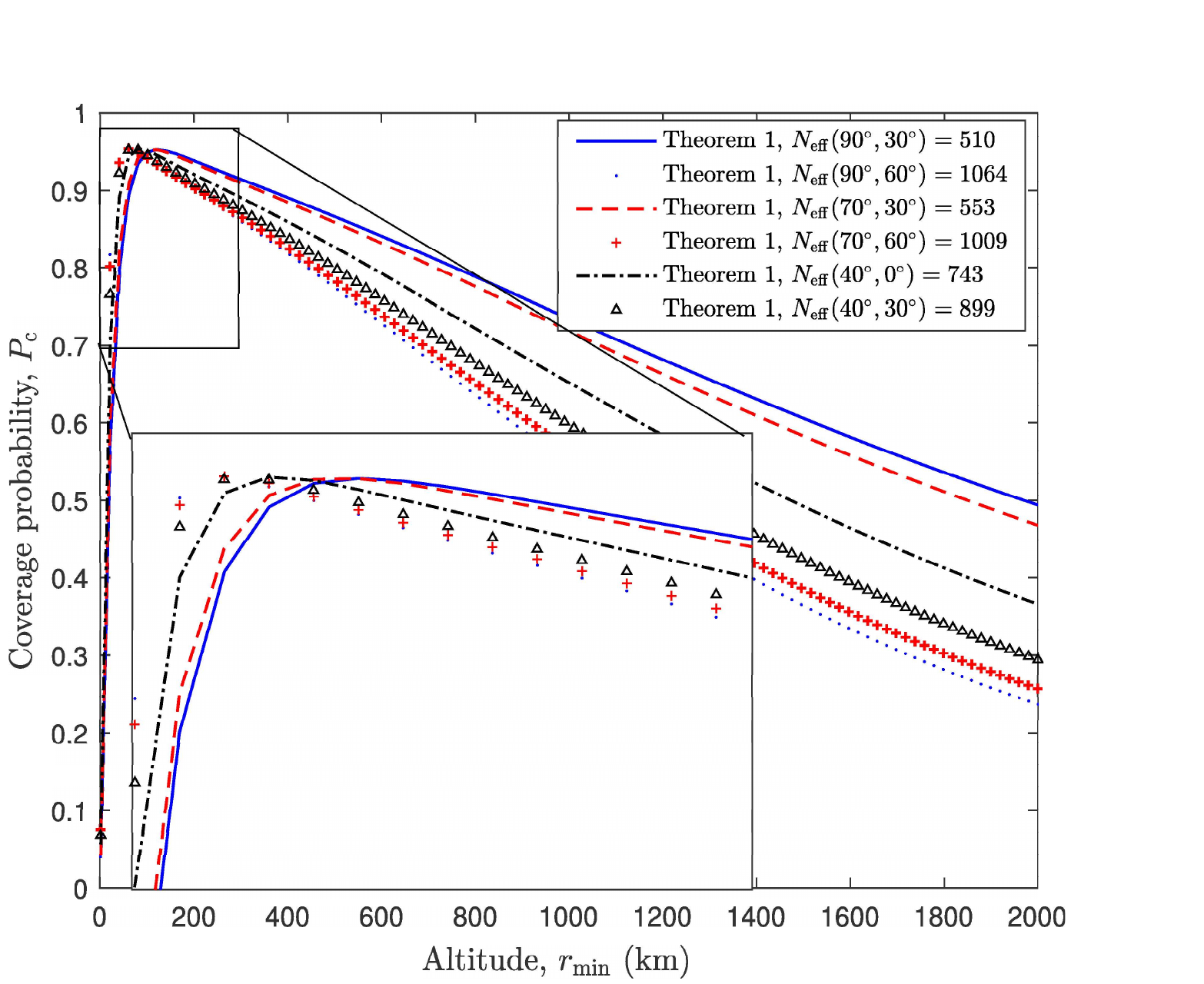}
    \caption{Coverage probability for different altitudes. Curves are plotted using Theorem 1 for different Walker constellations in $\phiu=0^\circ$, $30^\circ$, and $60^\circ$ assuming $N=\Neff(\phii,\phiu)$.}
    \label{fig:12}
\end{figure}

\appendix

\subsection{Proof of Lemma~1}
From basic geometry, the CDF of the surface area of the shaded spherical cap $A_{\text{cap}}$, {formed by any satellite at distance $R$ from the user}, in Fig.~2 is $F_{A_{\text{cap}}}\left(A_i\right)=\frac{A_i}{4\pi (\rEarth+\rmin)^2}$. To find the distribution of $R$, we need to find a relationship between $A_{\text{cap}}$ and $R$, when know that
\begin{align}
\label{eq:4}
A_{\text{cap}} &= \pi \left(a^2+b^2\right),\\
\label{eq:5}
R^2 &= \left(\rmin-b\right)^2+a^2,
\end{align}
{
for which $a$ and $b$ are given in Fig.~2.}
Combining \eqref{eq:4} and \eqref{eq:5}, we have   
\begin{align}
\label{eq:6}
\nonumber
R^2 &= \rmin^2-2\rmin b+b^2+a^2=\rmin^2-2\rmin b+\frac{A_{\text{cap}}}{\pi}\\
&= \rmin^2-2\rmin (\rEarth+\rmin)\left(1-\cos\theta\right)+\frac{A_{\text{cap}}}{\pi}.
\end{align}
Using the formula for the surface area of a spherical cap, i.e., $A_{\text{cap}}=2\pi (\rEarth+\rmin)^2(1-\cos\theta)$, we obtain
\begin{align}
\label{eq:7}
\nonumber
R^2 &= \rmin^2-2\rmin (\rEarth+\rmin)\left(\frac{A_{\text{cap}}}{2\pi (\rEarth+\rmin)^2}\right) + \frac{A_{\text{cap}}}{\pi}\\
&= \rmin^2+\frac{A_{\text{cap}}}{\pi}\left(1-\frac{\rmin}{\rEarth+\rmin}\right).
\end{align}
The CDF can then be deduced as follows:
\begin{align}
\label{eq:8}
\nonumber
\Prob\left(R<r\right)
&=\Prob\left(R^2<r^2\right)\\\nonumber
&=\Prob\left(\rmin^2+\frac{A_{\text{cap}}}{\pi}\left(1-\frac{\rmin}{\rEarth+\rmin}\right)<r^2\right)\\\nonumber
&=\Prob\left(A_{\text{cap}}<\frac{\pi \left(r^2-\rmin^2\right)}{1-\frac{\rmin}{\rEarth+\rmin}}\right)\\\nonumber
&=\frac{\pi \left(r^2-\rmin^2\right)}{\left(1-\frac{\rmin}{\rEarth+\rmin}\right)4\pi \left(\rEarth+\rmin\right)^2}\\
&=\frac{r^2-\rmin^2}{4\rEarth(\rEarth+\rmin)}.
\end{align}
Finally, the corresponding PDF can be derived by differentiating \eqref{eq:8} with respect to $r$.

\begin{figure}
         \includegraphics[trim = 5mm 0mm 6cm 17.5cm, clip,width=\textwidth]{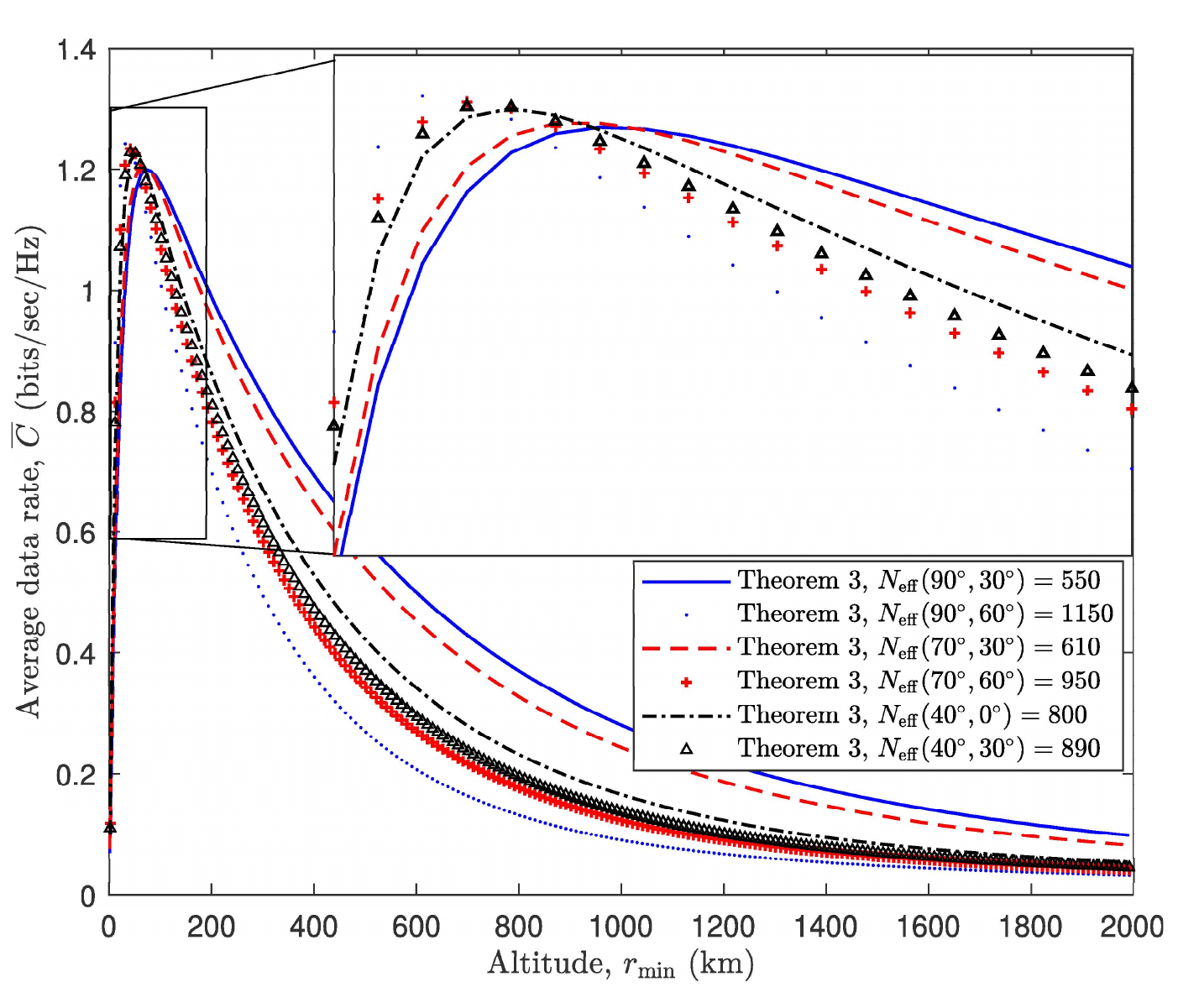}
       \caption{Average rate for different altitudes. Curves are plotted using Theorem 3 for different Walker constellations in $\phiu=0^\circ$, $30^\circ$, and $60^\circ$ assuming $N=\Neff(\phii,\phiu)$.}
      \label{fig:13}
\end{figure}

\subsection{Proof of Lemma~5}
In this Appendix, we provide the proof to the expression of Laplace function for general fading channels. Using the definition of the Laplace transform yields

\begin{align}
\label{eq:23}
\nonumber
\mathcal{L}_{I}(s)&\triangleq\mathbb{E}_{I}\left[e^{-sI}\right]=\mathbb{E}_{\NI, R_n, G_n}\left[\exp{\left(-s\sum_{n=1}^{\NI}\ppi G_nR_n^{-\alpha}\right)}\right]\\\nonumber
&=\mathbb{E}_{\NI, R_n, G_n}\left[\prod_{n=1}^{\NI}\exp{\left(-s\ppi G_nR_n^{-\alpha}\right)}\right]\\\nonumber
&\stackrel{(a)}= \mathbb{E}_{\NI, R_n}\left[\prod_{n=1}^{\NI}\mathbb{E}_{G_n}\left[\exp{\left(-s\ppi G_nR_n^{-\alpha}\right)}\right]\right]\\\nonumber
&\stackrel{(b)}= \mathbb{E}_{\NI}\Big[\prod_{n=1}^{\NI}\frac{2}{\rmax^4/\rmin^2-r_0^2}\\\nonumber
&\times \int_{r_0}^{\rmax}\mathbb{E}_{G_n}\left[\exp{\left(-s\ppi G_n r_n^{-\alpha}\right)}\right]r_n\,dr_n\Big]\\\nonumber
&\stackrel{(c)}= \sum_{\nI=1}^{\frac{N}{K}-1} \Bigg(\binom{\frac{N}{K}-1}{\nI}\PI^{\nI}(1-\PI)^{\frac{N}{K}-1-\nI}\Big(\frac{2}{\rmax^4/\rmin^2-r_0^2}\\\nonumber
&\times {\int_{r_0}^{\rmax}\mathbb{E}_{G_n}\left[\exp{\left(-s\ppi G_n r_n^{-\alpha}\right)}\right]\,r_n\,dr_n\Big)}^{\nI}\Bigg)\\\nonumber
&\stackrel{(d)}= \sum_{\nI=1}^{\frac{N}{K}-1} \Bigg(\binom{\frac{N}{K}-1}{\nI}\PI^{\nI}(1-\PI)^{\frac{N}{K}-1-\nI}\\
&\times {\Big(\frac{2}{\rmax^4/\rmin^2-r_0^2}\int_{r_0}^{\rmax}\mathcal{L}_{G_n}\left(s\ppi r_n^{-\alpha}\right)\,r_n\,dr_n\Big)}^{\nI}\Bigg)
\end{align}
where (a) follows from the i.i.d.\ distribution of $G_n$ and its further independence from $\NI$ and $R_n$,  (b) is obtained using the interfering distance distribution from \eqref{eq:14}, (c) is the averaging over the binomial random variable $\NI$ with the success probability $\PI$, which is given in Lemma~4, {and (d) is the substitution from the definition of Laplace function. }

\subsection{Proof of Theorem 3}
We only provide the proof for \eqref{eq:33} herein, while the proof for \eqref{eq:37} follows the same approach as well. In particular, 
\begin{align}
\label{eq:24}
\nonumber
&\mathbb{E}_{I,G_0,R_0}\left[\log_2\left(1+\SINR\right)|\NI>0\right]\\\nonumber
&=c_0\int_{\rmin}^{\rmax}\mathbb{E}\left[\ln\left(1+\frac{\ps G_0r_0^{-\alpha}}{I+\sigma^2}\right)\Big| \NI>0\right]\\\nonumber
&\hspace{45pt}\times\left(1-\frac{r_0^2-\rmin^2}{4\rEarth(\rEarth+\rmin)}\right)^{N-1}r_0\,dr_0\\\nonumber
&\stackrel{(a)}=c_0\int_{\rmin}^{\rmax}\int_{t>0}\Prob\left(\ln\left(1+\frac{\ps G_0r_0^{-\alpha}}{I+\sigma^2}\right)>t\Big| \NI>0\right)dt\\\nonumber
&\hspace{65pt}\times\left(1-\frac{r_0^2-\rmin^2}{4\rEarth(\rEarth+\rmin)}\right)^{N-1}r_0\,dr_0\\\nonumber
&=c_0\int_{\rmin}^{\rmax}\int_{t>0}\Prob\left(G_0>\frac{r_0^\alpha}{\ps }\left(\sigma^2+I\right)\left(e^t-1\right)\Big| \NI>0\right)dt\\\nonumber
&\hspace{65pt}\times\left(1-\frac{r_0^2-\rmin^2}{4\rEarth(\rEarth+\rmin)}\right)^{N-1}r_0\,dr_0\\\nonumber
&=c_0\int_{\rmin}^{\rmax}\int_{t>0}\mathbb{E}_I\left[e^{-\frac{r_0^\alpha}{\ps }\left(I+\sigma^2\right)\left(e^t-1\right)}\Big | \NI>0\right]dt\\\nonumber
&\hspace{65pt}\times\left(1-\frac{r_0^2-\rmin^2}{4\rEarth(\rEarth+\rmin)}\right)^{N-1}r_0\,dr_0\\\nonumber
&=c_0\int_{\rmin}^{\rmax}\int_{t>0}e^{-\frac{r_0^\alpha}{\ps }\sigma^2 \left(e^t-1\right)}\mathbb{E}_I\left[e^{-\frac{r_0^\alpha}{\ps } \left(e^t-1\right)I}\right]dt\\
&\hspace{65pt}\times\left(1-\frac{r_0^2-\rmin^2}{4\rEarth(\rEarth+\rmin)}\right)^{N-1}r_0\,dr_0,
\end{align}
where $c_0=\frac{N}{2 \ln(2)\rEarth(\rEarth+\rmin)}$ and (a) follows from the fact that for a positive random variable $X$, $\mathbb{E}\left[X\right]=\int_{t>0}\Prob\left(X>t\right)dt$. 

\subsection{Proof of Theorem 4}
Assuming $G_n=1$, we have
\begin{align}
\label{eq:30}
\nonumber
&\mathbb{E}_{I,R_0}\left[\log_2\left(1+\SINR\right)|\NI>0\right]\\\nonumber
&=c_0\int_{\rmin}^{\rmax}\mathbb{E}_{I}\left[\ln\left(1+\frac{\ps r_0^{-\alpha}}{I+\sigma^2}\right)\Big|\NI>0\right]\\\nonumber
&\hspace{65pt}\times \left(1-\frac{r_0^2-\rmin^2}{4\rEarth(\rEarth+\rmin)}\right)^{N-1}r_0\,dr_0\\\nonumber
&= c_0\int_{\rmin}^{\rmax}\int_{t>0}\Prob\left(\ln\left(1+\frac{\ps r_0^{-\alpha}}{I+\sigma^2}\right)>t\Big|\NI>0\right)dt\\\nonumber
&\hspace{65pt}\times\left(1-\frac{r_0^2-\rmin^2}{4\rEarth(\rEarth+\rmin)}\right)^{N-1}r_0\,dr_0\\\nonumber
&= c_0\int_{\rmin}^{\rmax}\int_{t>0}\Prob\left(\ps >r_0^\alpha\left(\sigma^2+I\right)\left(e^t-1\right)dt\Big| \NI>0\right)\\\nonumber
&\hspace{65pt}\times\left(1-\frac{r_0^2-\rmin^2}{4\rEarth(\rEarth+\rmin)}\right)^{N-1}r_0\,dr_0\\\nonumber
&= c_0\int_{\rmin}^{\rmax}\int_{t>0}\Prob\left(0<I<\frac{\ps }{r_0^\alpha\left(e^t-1\right)}-\sigma^2\right)dt\\
&\hspace{65pt}\times\left(1-\frac{r_0^2-\rmin^2}{4\rEarth(\rEarth+\rmin)}\right)^{N-1}r_0\,dr_0,
\end {align}
where $c_0=\frac{N}{2 \ln(2)\rEarth(\rEarth+\rmin)}$. Following the same approach as in \eqref{eq:34}--\eqref{eq:34-3} and using Parseval--Plancherel formula will result in \eqref{eq:2222}. The proof for \eqref{eq:31} can be obtained easily using the same principles as the above and only by substituting $I=0$.

\bibliography{main_TwoColumn_revised} 

\begin{thebibliography}{10}
\providecommand{\url}[1]{#1}
\csname url@samestyle\endcsname
\providecommand{\newblock}{\relax}
\providecommand{\bibinfo}[2]{#2}
\providecommand{\BIBentrySTDinterwordspacing}{\spaceskip=0pt\relax}
\providecommand{\BIBentryALTinterwordstretchfactor}{4}
\providecommand{\BIBentryALTinterwordspacing}{\spaceskip=\fontdimen2\font plus
\BIBentryALTinterwordstretchfactor\fontdimen3\font minus
  \fontdimen4\font\relax}
\providecommand{\BIBforeignlanguage}[2]{{%
\expandafter\ifx\csname l@#1\endcsname\relax
\typeout{** WARNING: IEEEtran.bst: No hyphenation pattern has been}%
\typeout{** loaded for the language `#1'. Using the pattern for}%
\typeout{** the default language instead.}%
\else
\language=\csname l@#1\endcsname
\fi
#2}}
\providecommand{\BIBdecl}{\relax}
\BIBdecl

\bibitem{1}
F.~Vatalaro, G.~E. Corazza, C.~Caini, and C.~Ferrarelli, ``Analysis of {LEO},
  {MEO}, and {GEO} global mobile satellite systems in the presence of
  interference and fading,'' \emph{IEEE Journal on Selected Areas in
  Communications}, vol.~13, no.~2, pp. 291--300, Feb. 1995.

\bibitem{4}
A.~Ganz, Y.~Gong, and B.~Li, ``Performance study of low {Earth}-orbit satellite
  systems,'' \emph{IEEE Transactions on Communications}, vol.~42, no. 234, pp.
  1866--1871, Feb. 1994.

\bibitem{38}
H.~M. {Mourad}, A.~A.~M. {Al-Bassiouni}, S.~S. {Emam}, and E.~K. {Al-Hussaini},
  ``Generalized performance evaluation of low {Earth} orbit satellite
  systems,'' \emph{IEEE Communications Letters}, vol.~5, no.~10, pp. 405--407,
  Oct. 2001.

\bibitem{5}
G.~Ruiz, T.~L. Doumi, and J.~G. Gardiner, ``Teletraffic analysis and simulation
  for nongeostationary mobile satellite systems,'' \emph{IEEE Transactions on
  Vehicular Technology}, vol.~47, no.~1, pp. 311--320, Feb. 1998.

\bibitem{8}
I.~Ali, N.~Al-Dhahir, and J.~E. Hershey, ``Predicting the visibility of {LEO}
  satellites,'' \emph{IEEE Transactions on Aerospace and Electronic Systems},
  vol.~35, no.~4, pp. 1183--1190, Oct. 1999.

\bibitem{24}
M.~Haenggi, \emph{Stochastic geometry for wireless networks}.\hskip 1em plus
  0.5em minus 0.4em\relax Cambridge University Press, 2012.

\bibitem{blaszczyszyn2018stochastic}
B.~B{\l}aszczyszyn, M.~Haenggi, P.~Keeler, and S.~Mukherjee, \emph{Stochastic
  geometry analysis of cellular networks}.\hskip 1em plus 0.5em minus
  0.4em\relax Cambridge University Press, 2018.

\bibitem{25}
M.~{Haenggi}, J.~G. {Andrews}, F.~{Baccelli}, O.~{Dousse}, and
  M.~{Franceschetti}, ``Stochastic geometry and random graphs for the analysis
  and design of wireless networks,'' \emph{IEEE Journal on Selected Areas in
  Communications}, vol.~27, no.~7, pp. 1029--1046, Sep. 2009.

\bibitem{23}
H.~ElSawy, E.~Hossain, and M.~Haenggi, ``Stochastic geometry for modeling,
  analysis, and design of multi-tier and cognitive cellular wireless networks:
  A survey,'' \emph{IEEE Communications Surveys and Tutorials}, vol.~15, no.~3,
  pp. 996--1019, Jun. 2013.

\bibitem{9}
J.~G. Andrews, F.~Baccelli, and R.~K. Ganti, ``A tractable approach to coverage
  and rate in cellular networks,'' \emph{IEEE Transactions on Communications},
  vol.~59, no.~11, pp. 3122--3134, Nov. 2011.

\bibitem{10}
H.~S. Dhillon, R.~K. Ganti, F.~Baccelli, and J.~G. Andrews, ``Modeling and
  analysis of k-tier downlink heterogeneous cellular networks,'' \emph{IEEE
  Journal on Selected Areas in Communications}, vol.~30, no.~3, pp. 550--560,
  Apr. 2012.

\bibitem{11}
D.~Cao, S.~Zhou, and Z.~Niu, ``Optimal combination of base station densities
  for energy-efficient two-tier heterogeneous cellular networks,'' \emph{IEEE
  Transactions on Wireless Communications}, vol.~12, no.~9, pp. 4350--4362,
  Sep. 2013.

\bibitem{12}
H.~S. Dhillon, T.~D. Novlan, and J.~G. Andrews, ``Coverage probability of
  uplink cellular networks,'' in \emph{Proc. IEEE Global Communications
  Conference (GLOBECOM)}, Dec. 2012.

\bibitem{28}
M.~{Afshang} and H.~S. {Dhillon}, ``Fundamentals of modeling finite wireless
  networks using binomial point process,'' \emph{IEEE Transactions on Wireless
  Communications}, vol.~16, no.~5, pp. 3355--3370, May 2017.

\bibitem{19}
Z.~Pan and Q.~Zhu, ``Modeling and analysis of coverage in {3-D} cellular
  networks,'' \emph{IEEE Communications Letters}, vol.~19, no.~5, pp. 831--834,
  May 2015.

\bibitem{20}
V.~V. Chetlur and H.~S. Dhillon, ``Downlink coverage analysis for a finite
  {3-D} wireless network of unmanned aerial vehicles,'' \emph{IEEE Transactions
  on Communications}, vol.~65, no.~10, pp. 4543--4558, Oct. 2017.

\bibitem{41}
S.~{Srinivasa} and M.~{Haenggi}, ``Distance distributions in finite uniformly
  random networks: Theory and applications,'' \emph{IEEE Transactions on
  Vehicular Technology}, vol.~59, no.~2, pp. 940--949, Feb. 2010.

\bibitem{42}
J.~{Guo}, S.~{Durrani}, and X.~{Zhou}, ``Outage probability in
  arbitrarily-shaped finite wireless networks,'' \emph{IEEE Transactions on
  Communications}, vol.~62, no.~2, pp. 699--712, Feb. 2014.

\bibitem{43}
Z.~{Khalid} and S.~{Durrani}, ``Distance distributions in regular polygons,''
  \emph{IEEE Transactions on Vehicular Technology}, vol.~62, no.~5, pp.
  2363--2368, Jun. 2013.

\bibitem{44}
J.~{Guo}, S.~{Durrani}, and X.~{Zhou}, ``Performance analysis of
  arbitrarily-shaped underlay cognitive networks: Effects of secondary user
  activity protocols,'' \emph{IEEE Transactions on Communications}, vol.~63,
  no.~2, pp. 376--389, Feb. 2015.

\bibitem{45}
X.~{Wang}, H.~{Zhang}, Y.~{Tian}, and V.~C.~M. {Leung}, ``Modeling and analysis
  of aerial base station-assisted cellular networks in finite areas under {LoS}
  and {NLoS} propagation,'' \emph{IEEE Transactions on Wireless
  Communications}, vol.~17, no.~10, pp. 6985--7000, Oct. 2018.

\bibitem{30}
M.~{Sellathurai}, S.~{Vuppala}, and T.~{Ratnarajah}, ``User selection for
  multi-beam satellite channels: A stochastic geometry perspective,'' in
  \emph{Proc. Asilomar Conference on Signals, Systems and Computers}, Nov.
  2016.

\bibitem{39}
S.~{Zhang} and J.~{Liu}, ``Analysis and optimization of multiple unmanned
  aerial vehicle-assisted communications in post-disaster areas,'' \emph{IEEE
  Transactions on Vehicular Technology}, vol.~67, no.~12, pp. 12\,049--12\,060,
  Dec 2018.

\bibitem{31}
O.~Y. {Kolawole}, S.~{Vuppala}, M.~{Sellathurai}, and T.~{Ratnarajah}, ``On the
  performance of cognitive satellite--terrestrial networks,'' \emph{IEEE
  Transactions on Cognitive Communications and Networking}, vol.~3, no.~4, pp.
  668--683, Dec. 2017.

\bibitem{6}
Y.~Seyedi and S.~M. Safavi, ``On the analysis of random coverage time in mobile
  {LEO} satellite communications,'' \emph{IEEE Communications Letters},
  vol.~16, no.~5, pp. 612--615, May 2012.

\bibitem{34}
S.~A. {Banani}, A.~W. {Eckford}, and R.~S. {Adve}, ``Analyzing the impact of
  access point density on the performance of finite-area networks,'' \emph{IEEE
  Transactions on Communications}, vol.~63, no.~12, pp. 5143--5161, Dec. 2015.

\bibitem{40}
\BIBentryALTinterwordspacing
N.~K. Lyras, C.~I. Kourogiorgas, and A.~D. Panagopoulos, ``Cloud free line of
  sight prediction modeling for low {Earth} orbit optical satellite networks,''
  2019. [Online]. Available: \url{https://doi.org/10.1117/12.2535971.}
\BIBentrySTDinterwordspacing

\bibitem{table}
I.~S. Gradshteyn and I.~M. Ryzhik, \emph{Table of Integrals, Series, and
  Products}.\hskip 1em plus 0.5em minus 0.4em\relax Academic Press, 2014.

\bibitem{37}
F.~Baccelli and B.~B{\l}aszczyszyn, ``On a coverage process ranging from the
  {Boolean} model to the {Poisson}--{Voronoi} tessellation with applications to
  wireless communications,'' \emph{Advances in Applied Probability}, vol.~33,
  no.~2, pp. 293--323, Jun. 2001.

\bibitem{29}
------, ``Stochastic geometry and wireless networks: Volume {II}
  applications,'' \emph{Foundations and Trends{\textregistered} in Networking},
  vol.~4, no. 1--2, pp. 1--312, Jan. 2010.

\end{thebibliography}
\bibliographystyle{IEEEtran}

\begin{IEEEbiography}[{\includegraphics[width=1in,height=1.25in,clip,keepaspectratio]{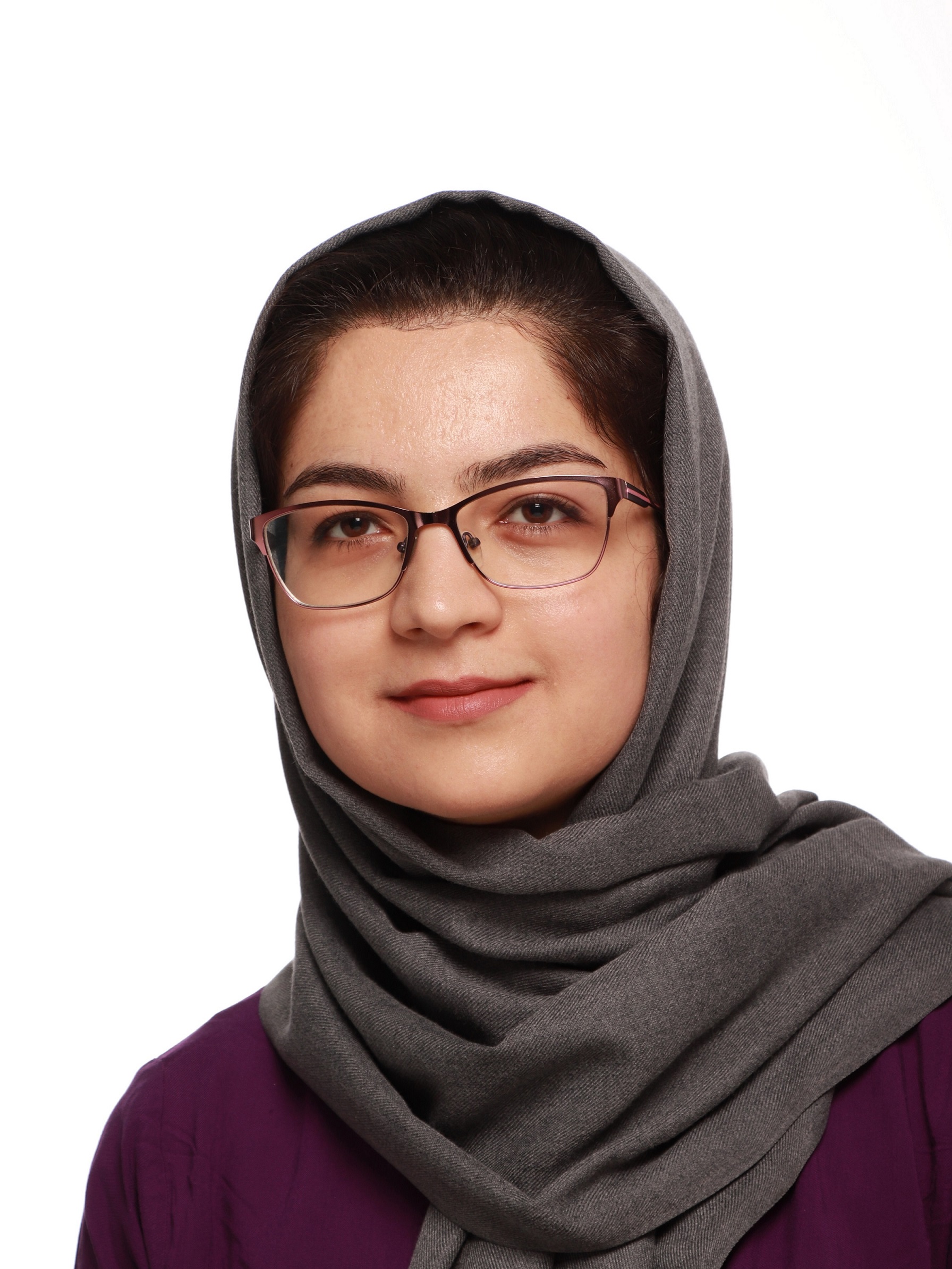}}]{Niloofar Okati}
received the B.E.\ degree in electrical engineering from Shiraz University of Technology, Iran, in 2013, and M.Sc.\ degree in communications engineering from Iran University of Science and Technology (IUST), Iran, in 2016. She is currently a Ph.D.\ researcher at the Faculty of Information Technology and Communication Sciences, Tampere University, Finland. Her research interests include stochastic geometry, wireless communication and satellite networks. She was one of the 200 young researchers in mathematics and computer science selected worldwide to participate in the 7th Heidelberg Laureate Forum (HLF), in 2019. 
\end{IEEEbiography}

\begin{IEEEbiography}[{\includegraphics[width=1in,height=1.25in,clip,keepaspectratio]{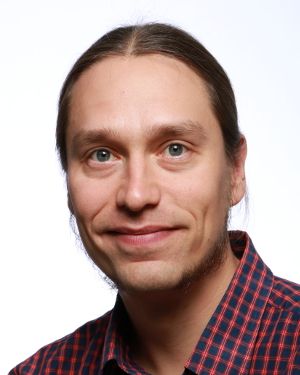}}]{Taneli Riihonen}(S'06--M'14)
received the D.Sc.\ degree in electrical engineering (with distinction) from Aalto University, Helsinki, Finland, in August 2014. He is currently an Assistant Professor (tenure track) at the Faculty of Information Technology and Communication Sciences, Tampere University, Finland. He held various research positions at Aalto University School of Electrical Engineering from September 2005 through December 2017. He was a Visiting Associate Research Scientist and an Adjunct Assistant Professor at Columbia University in the City of New York, USA, from November 2014 through December 2015. He has been nominated eleven times as an Exemplary/Top Reviewer of various IEEE journals and is serving as an Editor for \textsc{IEEE Wireless Communications Letters} since May 2017. He has previously served as an Editor for \textsc{IEEE Communications Letters} from October 2014 through January 2019. He received the Finnish technical sector's award for the best doctoral dissertation in 2014 and the EURASIP Best PhD Thesis Award 2017. His research activity is focused on physical-layer OFDM(A), multiantenna, relaying and full-duplex wireless techniques with current interest in the evolution of beyond 5G systems.
\end{IEEEbiography}

\begin{IEEEbiography}[{\includegraphics[width=1in,height=1.25in,clip,keepaspectratio]{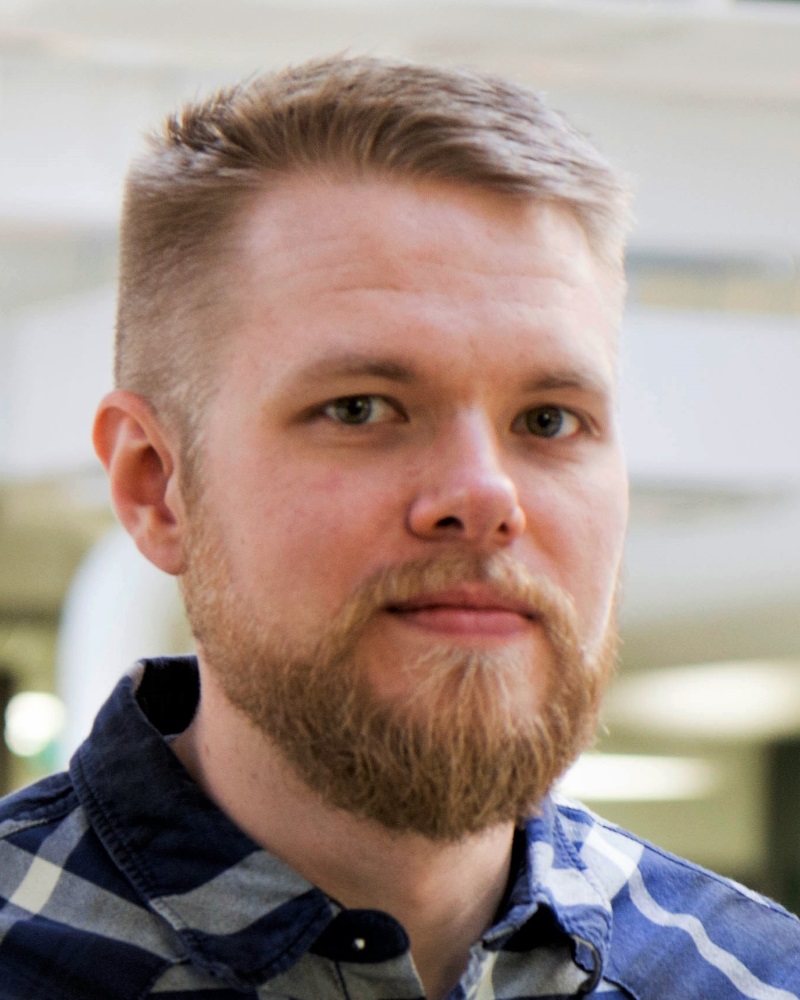}}]{Dani Korpi} received the M.Sc. and D.Sc. degrees (Hons.) in communications engineering and electrical engineering from Tampere University of Technology, Finland, in 2014 and 2017, respectively. Currently, he is a Senior Specialist with Nokia Bell Labs in Espoo, Finland. His doctoral thesis received an award for the best dissertation of the year in Tampere University of Technology, as well as the Finnish technical sector's award for the best doctoral dissertation in 2017. His research interests include inband full-duplex radios, machine learning for wireless communications, and beyond 5G radio systems.
\end{IEEEbiography}

\begin{IEEEbiography}[{\includegraphics[width=1in,height=1.25in,clip,keepaspectratio]{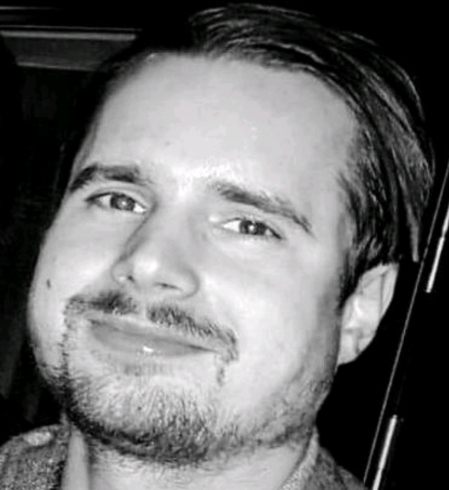}}]{Ilari Angervuori}  received the B.Sc.\ and M.Sc.\ degrees from the University of Helsinki, Finland, in 2016 and 2018. Since 2018, he has been with the School of Electrical Engineering at Department of Signal Processing and Acoustics, Aalto University, where he has been a Doctoral Candidate since 2019. His research interests are in the areas of satellite communications and wireless communication. He has a strong background in applied mathematics and applying stochastic geometry into areas of engineering is in his main focus.
\end{IEEEbiography}

\begin{IEEEbiography}[{\includegraphics[width=1in,height=1.25in,clip,keepaspectratio]{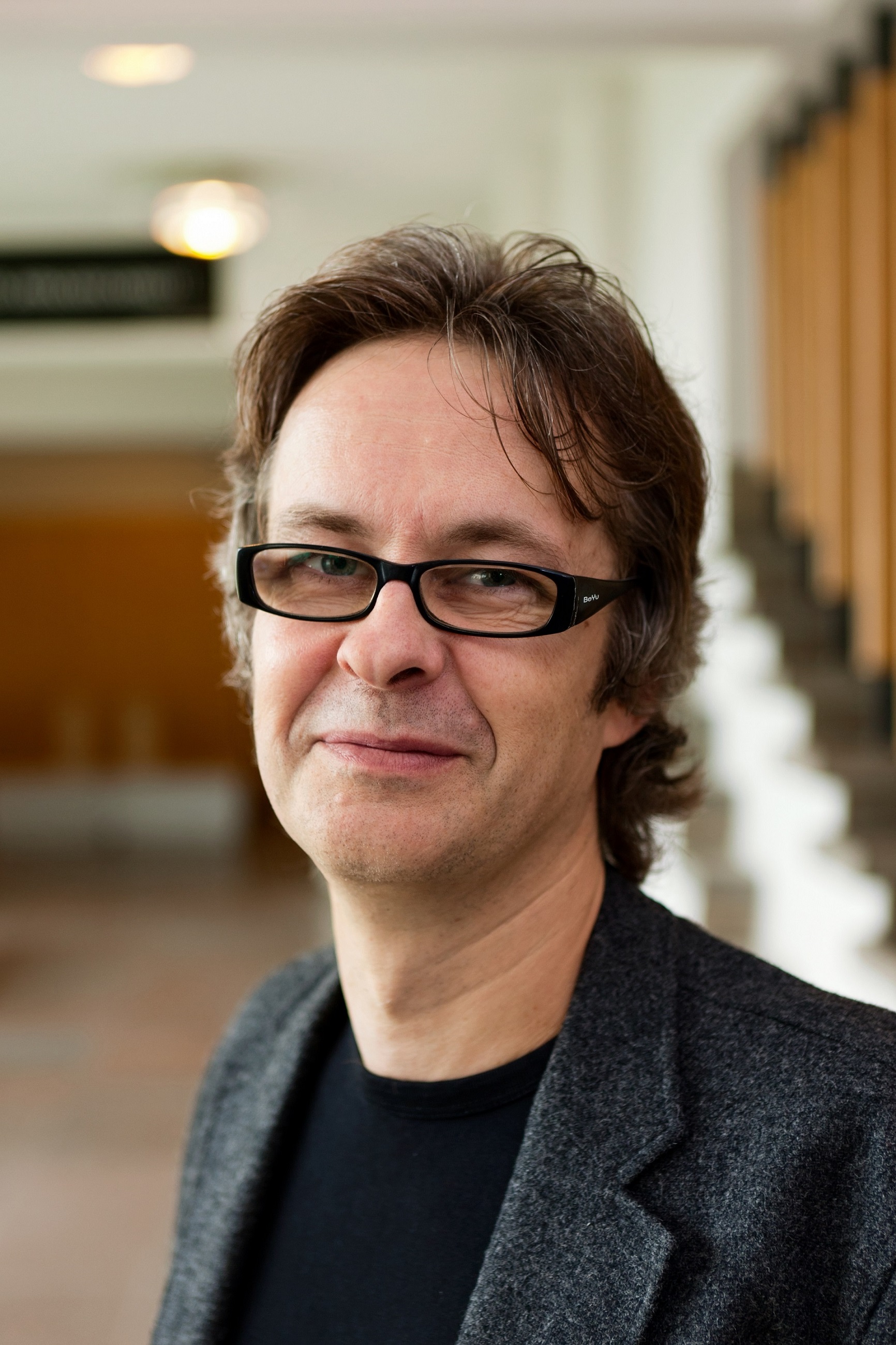}}]{Risto Wichman} received his M.Sc.\  and D.Sc.\ (Tech.) degrees in digital signal processing from Tampere University of Technology, Finland, in 1990 and 1995, respectively.  From 1995 to 2001, he worked at Nokia Research Center as a senior research engineer.  In 2002, he joined Department of Signal Processing and Acoustics, Aalto University School of Electrical Engineering, Finland, where he is a full professor since 2008.  His research interests include signal processing techniques for wireless communication systems.
\end{IEEEbiography}

\end{document}